\documentclass[twocolumn,letterpaper,journal]{IEEEtran}

\usepackage{amsmath,amssymb}
\usepackage{mathtools}
\usepackage{bm}

\usepackage{amsthm} 
\usepackage{subfig}
\usepackage{xcolor}

\usepackage{enumitem}

\newtheorem{lemma}{Lemma}
\newtheorem{proposition}{Proposition}
\newtheorem{theorem}{Theorem}
\newtheorem{corollary}{Corollary}
\newtheorem{remark}{Remark}
\newtheorem{assumption}{Assumption}

\newcommand{\R}{\mathbb{R}}

\newcommand{\one}{\mathbf{1}}

\title{Feedforward Density-Driven Optimal Control (D$^2$OC) for Tracking Time-Varying Distributions}

\author{Kooktae Lee$^{1}$ and Julian Martinez$^{1}$
	\thanks{*This work was supported by NSF CAREER Grant CMMI-DCSD-2638508.}
	\thanks{$^{1}$Kooktae Lee and Julian Martinez are with the Department of Mechanical and Aerospace Engineering, Texas Tech University, Lubbock, TX 79409, USA, email: kooktae.lee@ttu.edu, mar85781@ttu.edu.}
}

\begin{document}

	\maketitle
	\thispagestyle{empty}
	\pagestyle{empty}

\begin{abstract}
Tracking a time-varying spatial distribution poses a fundamental challenge for multi-agent systems, as agents often lag behind a moving target when reacting solely to the current target configuration. This paper develops a feedforward extension of Density-Driven Optimal Control (D$^2$OC) for dynamic distribution tracking with general discrete-time linear time-invariant agents. The proposed framework predicts reference sample trajectories along a known velocity field and incorporates these anticipated locations into the finite-horizon control problem. A scalar factor explicitly adjusts the feedforward action to account for predicted reference displacement. We establish an explicit receding-horizon control policy and provide a formal performance analysis that quantifies the tracking impact of prediction limitations. Furthermore, under a nominal stability condition, we derive explicit ultimate bounds on agent tracking errors and local Wasserstein distances in the presence of bounded reference drift, discretization errors, and transport perturbations. Monte Carlo simulations with linearized quadrotor dynamics demonstrate that incorporating feedforward compensation significantly mitigates tracking lag and enhances spatial distribution matching. These results offer formal analytical guarantees for predictive distribution tracking while preserving the computational efficiency and local optimization structure of D$^2$OC.
\end{abstract}

\begin{IEEEkeywords}
Multi-agent systems, optimal transport, density-driven optimal control, time-varying distribution tracking, feedforward control.
\end{IEEEkeywords}

\section{Introduction}
\label{sec:introduction}

Multi-agent robots are increasingly used in applications such as environmental plume monitoring \cite{lewis2021comprehensive,wang2019dynamic}, wildfire monitoring \cite{bailon2022real}, urban surveillance \cite{leong2021intelligent}, and precision agriculture \cite{tokekar2016sensor,seo2025tcst}, where the importance of observations varies across the region. In such settings, the team should distribute its sensing or actuation effort according to a spatial density representing where observations are most valuable. Because this density can evolve over time, for example as a plume is transported or a fire front advances, effective coverage requires accounting for temporal variations, individual robot dynamics, and finite operational resources.

A foundational approach is to partition the environment and guide each robot toward a representative point of its assigned region. Voronoi-based coverage control, for example, employs spatial partitions and density-weighted centroids \cite{du1999cvt,cortes2004coverage,cortes2005spatially,schwager2009decentralized}. Extensions have addressed time-varying environments by adapting density functions or domains \cite{lee2015timevarying,atinc2020swarm,xu2019timevarying}, incorporating online learning \cite{pratissoli2025distributed}, reinforcement learning for dynamic UAV coverage \cite{ma2024reinforcement}, and distributed coverage for constrained systems \cite{liu2024distributed}. Although these approaches incorporate temporal variation and practical constraints, they focus predominantly on tracking instantaneous configurations. This differs fundamentally from matching the \textit{time-averaged} spatial distribution of robot trajectories to a target density over an operating period. In persistent monitoring, the information gathered is the accumulation of sensing effort over time, and thus a configuration that is well placed at each instant can still leave regions under-observed over the whole mission, and instantaneous configurations are an insufficient performance metric.

Ergodic control addresses this time-averaged distribution-matching objective by designing trajectories whose temporal statistics approximate a target spatial distribution \cite{mathew2011metrics,mavrommati2017real,miller2013ergodic}. However, a fundamental distinction remains between asymptotic ergodicity and distribution matching over a finite operating period. Conventional ergodic metrics characterize performance in the infinite-time limit, whereas practical platforms have finite operational lifetimes. Even though extensions have addressed time-varying information maps, handling a changing target does not eliminate the gap between asymptotic convergence and finite-horizon performance \cite{pratissoli2025distributed,mavrommati2017real}. This motivates formulations that directly optimize time-averaged distribution matching over a finite prediction horizon while accounting for agent dynamics.

Optimal Transport (OT) provides a principled framework for comparing spatial distributions by evaluating mass transportation costs \cite{villani2009optimal}. While Eulerian continuum formulations describe density propagation via PDE-constrained velocity fields \cite{benamou2000computational}, they rely on infinite-agent approximations unsuited for discrete multi-robot teams. Density-Driven Optimal Control (D$^2$OC) overcomes this by embedding a time-averaged Wasserstein distribution-matching objective directly into a finite-horizon optimal control problem \cite{seo2025tcst,seo2025density}, enabling discrete agents to match a target distribution through empirical trajectories. Subsequent developments extended D$^2$OC to preserve communication connectivity \cite{lee2025connectivity}, account for stochasticity \cite{lee2026automatica}, support decentralized terminal distribution \cite{lee2026tac}, and enable AI-supported environmental mapping \cite{lee2026ai}.

Despite these advances, previous D$^2$OC formulations assume that reference samples remain stationary over the prediction horizon, limiting their ability to track time-varying distributions. Even with accurate local optimization, a controller optimizing motion toward current reference locations does not account for future displacement, making the resulting tracking lag structural. Similar lag phenomena appear in time-varying optimization \cite{simonetto2016time,simonetto2020timevarying} and online control \cite{li2019online}. The challenge is thus to incorporate predicted reference motion directly into the finite-horizon D$^2$OC objective while maintaining its tractable structure and establishing explicit analytical guarantees.

To address this gap, we predict the motion of reference samples along a known velocity field and incorporate their predicted locations into the local D$^2$OC problem. A scalar tuning parameter determines the fraction of predicted motion incorporated over the horizon: setting this parameter to zero recovers the stationary-reference protocol, whereas full feedforward compensation utilizes the fully advected samples. This tunable interpolation framework provides a robust design trade-off when the underlying velocity field is imperfectly known. 

Building upon this feedforward-augmented predictive structure, the contributions of this paper are threefold. First, we formulate a feedforward-augmented D$^2$OC method for tracking time-varying distributions with general discrete-time linear time-invariant (LTI) agents, propagating reference samples via a discrete approximation of the continuity equation. Second, we derive a closed-form receding-horizon control law and an exact expression for the excess horizon cost, proving that it scales quadratically with uncompensated motion. Third, under a verifiable nominal-stability condition, we establish explicit ultimate bounds for the agent tracking error and local Wasserstein error in the presence of bounded drift, discretization error, and transport perturbations. Simulations using linearized quadrotors validate these theoretical findings by illustrating the significant reduction in tracking lag and distribution-matching error.

	\section{Preliminaries}
	\label{sec:preliminaries}

	\noindent \textbf{Notation:}
	The sets of real numbers and integers are represented by \(\mathbb{R}\) and \(\mathbb{Z}\), respectively. We denote positive and non-negative integers as \(\mathbb{Z}_{>0}\) and \(\mathbb{Z}_{\geq 0}\). Vectors in \(n\)-dimensional Euclidean space are denoted by \(\mathbb{R}^n\). The Euclidean norm is written as $\|\cdot\|$. For a matrix $A$, $A^\top$ indicates its transpose. Identity and zero matrices are denoted by $\mathbf{I}_n$ and $\mathbf{0}$, with subscripts indicating their dimensions. Given a positive definite matrix \(R \succ 0\), the weighted norm is defined as \(\|U\|_R = \sqrt{U^\top R U}\). The vector $\one_T$ is the all-ones vector of length $T$, $\otimes$ is the Kronecker product, and $\lambda_{\min}(\cdot)$, $\lambda_{\max}(\cdot)$ are extreme eigenvalues of symmetric matrices. The horizon length $T$ appears as an exponent only in matrix powers $A^T$, and transposes are always written $^\top$.

    \medskip
	Consider a multi-agent system where the motion of each agent \(i\in\{1,2,\ldots, n_a\}\), $n_a\in\mathbb{Z}_{>0}$ is governed by discrete-time Linear Time-Invariant (LTI) dynamics:
	\begin{equation}
		x_i(k+1) = A_i x_i(k) + B_i u_i(k), \quad y_i(k) = C_i x_i(k), \label{eq:dyn}
	\end{equation}
	where \(x_i(k) \in \mathbb{R}^n\), \(u_i(k) \in \mathbb{R}^m\), and \(y_i(k) \in \mathbb{R}^d\) represent the state, control input, and measured output (position) at discrete-time step \(k\in\mathbb{Z}_{\geq 0}\), respectively. The matrices \(A_i, B_i, C_i\) are system matrices for agent $i$.

	\subsection{Optimal Transport and Wasserstein Metric}
	The objective of time-averaged density matching is formulated through the lens of OT theory \cite{villani2009optimal}. We quantify the discrepancy between distributions using the \(p\)-Wasserstein distance. For two discrete probability measures \(\rho\) and \(\rho_{\text{ref}}\), this distance is defined as:
	\begin{equation}
    \small
		\mathcal{W}_p(\rho,\rho_{\text{ref}}) = \left( \min_{\pi_{\ell j}} \sum_{\ell=1}^M \sum_{j=1}^N \pi_{\ell j}\, \|y_\ell - q_j\|^p \right)^{1/p}, \label{eq:W-dist}
	\end{equation}
	subject to $\pi_{\ell j}\geq 0$, $\sum_{j=1}^N \pi_{\ell j} = \alpha_\ell$, and $\sum_{\ell=1}^M \pi_{\ell j} = \beta_j$, where $\alpha_\ell$ and $\beta_j$ are the weights of the $M$ support points $y_\ell$ of $\rho$ and of the $N$ reference samples $q_j$ of $\rho_{\text{ref}}$, respectively, each set summing to one. The term \(\pi_{\ell j}\) represents the transport plan between the agent positions \(y_\ell\) and reference samples \(q_j\). Throughout this work, we focus on the quadratic case ($p=2$).

	In this framework, the \textit{empirical distribution} \(\rho(k)\) is generated by the collective positions of \(n_a\) agents over time, while the \textit{target distribution} \(\rho_{\text{ref}}\) is represented by a set of reference samples \(q_j\) that encode the desired spatial importance. For a stationary reference, these distributions are expressed using Dirac measures $\delta$ as:
	\begin{equation*}
    \small
		\rho(k) = \frac{1}{k+1}\sum_{t=0}^k \left( \frac{1}{n_a}\sum_{i=1}^{n_a}\delta_{y_i(t)} \right), \quad \rho_{\text{ref}} = \frac{1}{N}\sum_{j=1}^N \delta_{q_j}.
	\end{equation*}
	The fundamental goal of D$^2$OC is to minimize the distance between $\rho(k)$ and $\rho_{\text{ref}}$ measured via the Wasserstein metric while adhering to the physical and operational constraints of the swarm.

	\subsection{Standard D$^2$OC Protocol and its Limitations}
	\label{subsec:std_d2oc}
	The standard D$^2$OC framework \cite{seo2025density} addresses the computational intractability of global Wasserstein minimization by decomposing the problem into a decentralized, three-stage iterative cycle:

	\begin{enumerate}[leftmargin=0pt, itemindent=*, labelwidth=0pt, align=left]
		\item \textbf{Local Control Synthesis:} Each agent selects a subset of proximal, high-weight reference samples and computes the input $u_i(k)$ that minimizes the local Wasserstein distance between its future state and these samples, subject to the dynamics \eqref{eq:dyn}.

		\item \textbf{Recursive Weight Update:} Each target location carries a weight $\beta_j(k)$, initialized as $\beta_j(0)=1/N$. As an agent covers a sample, the corresponding weight is reduced, so that the swarm keeps a record of the matching progress and allocates effort to regions not yet covered.

		\item \textbf{Distributed Consensus:} Agents within communication range synchronize their weight maps by adopting the minimum observed weight of each sample, which prevents redundant coverage.
	\end{enumerate}

	While this protocol is effective for stationary distributions, the samples are frozen at their current locations over the prediction horizon. It therefore incurs a structural tracking lag once $\rho_{\text{ref}}$ evolves, which motivates the feedforward-augmented protocol proposed in this paper.

	\section{Time-Varying Reference Samples for D$^2$OC}
	\label{sec:time_varying_samples}
	To extend the D$^2$OC framework to non-stationary scenarios, the reference distribution must be treated as a time-evolving entity. We characterize this evolution through \emph{time-varying sample points} $q_j(k) \in \mathbb{R}^d$, $j=1,\dots,N$, which discretize the target density at each time step $k$ and evolve according to the underlying reference flow.

	\subsection{Discrete-Time Continuity Equation Approximation}

	The reference density \(\rho_{\text{ref}}(x,t)\) in continuous time $t$ is treated as a time-varying spatial distribution. Following the dynamic formulation of optimal transport \cite{benamou2000computational}, this evolution is governed by the continuity equation:
	\begin{equation}
		\frac{\partial \rho_{\text{ref}}}{\partial t} + \nabla \cdot (\rho_{\text{ref}} v) = 0,
		\label{eq:continuity_eq}
	\end{equation}
	where \(v(x,t)\) represents the velocity field underlying the reference flow, which is treated as an exogenous, known input rather than a quantity solved online.

	\begin{assumption}[Bounded Reference Velocity]
		\label{ass:vbound}
		The velocity field satisfies $\|v(x,t)\|\le \bar v$ for all $x$ and $t$, where $\bar v \in\mathbb{R}$ is some positive constant.
	\end{assumption}

	Integrating the velocity field along the characteristics of \eqref{eq:continuity_eq} with a forward Euler step of size $\Delta t$ (writing $v(x,k) := v(x, k\Delta t)$ for the reference velocity field sampled at the $k$-th instant) yields the Lagrangian update law (the method of characteristics for \eqref{eq:continuity_eq}):
	\begin{equation}
		q_j(k+1) = q_j(k) + \Delta t \, v(q_j(k), k), \quad j=1,\dots,N.
		\label{eq:sample_vf}
	\end{equation}
			At time $k$, the same Euler scheme with the known field defines the predicted sample positions over the prediction horizon $T \in \mathbb{Z}_{>0}$ for $h = 1, \dots, T$
	\begin{equation}
		q_j(k{+}h|k)=q_j(k{+}h{-}1|k)+\Delta t\, v\big(q_j(k{+}h{-}1|k),k{+}h{-}1\big),
		\label{eq:pred_samples}
	\end{equation}
	with $q_j(k|k)=q_j(k)$. By Assumption~\ref{ass:vbound}, $\|q_j(k{+}h|k)-q_j(k)\|\le h\Delta t\,\bar v$.

	\subsection{Impact on D$^2$OC Barycenter Computation}

	Recall from Section~\ref{subsec:std_d2oc} that each agent $i$ selects a proximal, high-weight subset $\mathcal{S}_i(k)$ at each time step. Then, the D$^2$OC is cast into tracking the local barycenter minimizing its Wasserstein distance to those samples, where the barycenter is calculated by
	\begin{equation}
		\bar q_i(k) := \frac{\sum_{j \in \mathcal{S}_i(k)} \pi_j(k) q_j(k)}{\sum_{j \in \mathcal{S}_i(k)} \pi_j(k)}.
		\label{eq:local_barycenter_time_varying}
	\end{equation}
	Note that $\pi_j(k)\ge 0$ is the transport plan to move the weights of samples in $\mathcal{S}_i(k)$ to agent $i$ such that $\sum_{j\in\mathcal S_i(k)}\pi_j(k)=\alpha_i$, where $\alpha_i := 1/M_i$ per step with $M_i$ representing the total number of operational steps available based on its allocated budget or battery life. In this way, the original D$^2$OC achieves an energy-aware distribution matching.

At each sampling instant $k$, agent $i$ sorts the available reference samples $q_j(k)$ in ascending order of distance from its current position. It sequentially depletes the remaining sample weights $\beta_j(k)$ against its mass capacity $\alpha_i$, collecting local weights $\pi_j(k)\le\beta_j(k)$ until the allocated sum satisfies
$
	\sum_{j \in \mathcal{S}_i(k)} \pi_j(k) = \alpha_i.
$

The local sample index set $\mathcal{S}_i(k)$ is thus defined as the collection of sample indices $j$ whose weights are assigned to agent $i$. This greedy assignment never double-counts target mass and avoids a global linear program. The selection $(\mathcal S_i(k),\pi_j(k))$ is held \emph{fixed over the prediction horizon}, while the samples move according to \eqref{eq:pred_samples}.

Unlike the stationary case, the samples $q_j(k)$ now evolve under \eqref{eq:sample_vf}, and thus $\bar q_i(k)$ becomes a moving target and pursuing its current position incurs an inherent tracking lag. This motivates the \emph{feedforward compensation}, which anticipates the barycenter's motion to reduce the structural lag. Using the predicted samples \eqref{eq:pred_samples} with the selection $(\mathcal S_i(k),\pi_j(k))$ held fixed over the horizon, the predicted local barycenter is
	\begin{equation}
		\bar q_i(k{+}h|k) := \frac{1}{\alpha_i}\sum_{j \in \mathcal{S}_i(k)} \pi_j(k)\, q_j(k{+}h|k),\quad h=1,\dots,T,
		\label{eq:pred_barycenter}
	\end{equation}
	with $\bar q_i(k|k)=\bar q_i(k)$.

\begin{remark}[Discretization and Total Perturbation]
\label{remark:discretization_error}
The forward Euler discretization in \eqref{eq:sample_vf} omits higher-order $\Delta t$ terms, and the re-selection of local samples at each step adds jitter, since re-selection can move the local barycenter by an amount comparable to the sample spread. In the stability analysis of Section~\ref{sec:analysis}, both effects are bundled into a single bounded total perturbation $\varpi_i(k)$ of the local barycenter, whose bound must account for both.
\end{remark}
\section{Formulation of the D$^2$OC Cost Function via Quadratic Programming}
\label{sec:qp_formulation}

To track the time-varying spatial distributions introduced in Section~\ref{sec:time_varying_samples}, we formulate a finite-horizon D$^2$OC problem over a prediction horizon $T$. At each sampling instant $k$, agent $i$ optimizes its predicted state trajectory over the horizon $T$ using a local subset of target samples, with only the first control action implemented before receding to the next instant.

\subsection{Local Wasserstein Tracking Objective}
\label{subsec:local_tracking_objective}

Given the local sample set $\mathcal{S}_i(k)$ and the weights $\pi_j(k)$, the local Wasserstein tracking objective for agent $i$ over the horizon $T$ is formulated as
\begin{equation}
	\mathcal{J}_i(k) = \sum_{h=1}^{T} \sum_{j \in \mathcal{S}_i(k)} \pi_j(k) \, \|y_i(k{+}h) - q_j(k{+}h|k)\|^2,
    \label{eq:J_i}
\end{equation}
subject to the agent physical dynamics \eqref{eq:dyn} and the local mass constraint $\sum_{j \in \mathcal{S}_i(k)} \pi_j(k) = \alpha_i$. Again, the symbol $\pi_j(k)$ acts as the local transport mass weight associated with sample $q_j$, directly coupling the agent's tracking objective with its remaining operational capacity and local target coverage.

Then, we develop the following lemma.

\begin{lemma}[Exact Wasserstein--Barycenter Decomposition]
\label{lem:W_barycenter_decomposition}
Let $y_i\in\R^d$, samples $q_j\in\R^d$, and weights $\pi_j\ge0$ with $\sum_j\pi_j=\alpha_i>0$, and let $\bar q_i=\frac1{\alpha_i}\sum_j\pi_jq_j$. Then
\begin{equation}
\sum_j\pi_j\|y_i-q_j\|^2
=
\alpha_i\|y_i-\bar q_i\|^2
+
\psi_i,\quad
\psi_i:=\sum_j\pi_j\|q_j-\bar q_i\|^2 ,
\label{eq:W_exact_decomposition}
\end{equation}
where $\psi_i$ is the weighted spatial spread of the samples.
\end{lemma}

\begin{proof}
Write $y_i-q_j=(y_i-\bar q_i)+(\bar q_i-q_j)$, expand the squared norm, multiply by $\pi_j$, and sum over $j$. This gives $\alpha_i\|y_i-\bar q_i\|^2+\psi_i+2(y_i-\bar q_i)^\top\sum_j\pi_j(\bar q_i-q_j)$, and the last sum vanishes by the definition of $\bar q_i$, which proves \eqref{eq:W_exact_decomposition}.
\end{proof}

Now, defining the stacked output, the stacked predicted barycenters, and the weighting matrix as
$
Y_i^k
:=
\begin{bmatrix}
y_i(k+1)^\top &
\cdots &
y_i(k+T)^\top
\end{bmatrix}^{\top}$,
$
\bar Q_i^{k}
:=
\begin{bmatrix}
\bar q_i(k+1|k)^\top &
\cdots &
\bar q_i(k+T|k)^\top
\end{bmatrix}^{\top}$,
and
$
\boldsymbol{\Omega}_i
:=
\sqrt{\alpha_i}\,\mathbf I_{dT},
$
the following local Wasserstein cost is provided.

\begin{proposition}[Quadratic Representation of the Local Wasserstein Cost]
\label{proposition:equiv}
Let the predicted local barycenters $\bar q_i(k+h|k)$ be given by \eqref{eq:pred_barycenter}. Then,
\begin{equation}
\sum_{h=1}^{T}
\mathcal W_i^2(h|k)
=
\left\|
\boldsymbol{\Omega}_i
\left(
Y_i^k-\bar Q_i^{k}
\right)
\right\|^2
+\operatorname{const.},
\label{eq:wasserstein_quadratic}
\end{equation}
where $\mathcal W_i^2(h|k)$ denotes the $h$-th summand of the inner sum in \eqref{eq:J_i} and $\operatorname{const.}:=\sum_{h=1}^T\psi_i(h|k)$ collects the weighted spreads of the predicted samples around their barycenters and is independent of $Y_i^k$.
\end{proposition}

\begin{proof}
Apply Lemma~\ref{lem:W_barycenter_decomposition} for each $h$ with $y=y_i(k{+}h)$, samples $q_j(k{+}h|k)$, and weights $\pi_j(k)$, and sum over $h$. The first terms add up to $\alpha_i\|Y_i^k-\bar Q_i^{k}\|^2=\|\boldsymbol\Omega_i(Y_i^k-\bar Q_i^{k})\|^2$.
\end{proof}

\subsection{Input and Terminal Regularization}

The finite-horizon Wasserstein objective \eqref{eq:wasserstein_quadratic} alone does not penalize the input effort or the state remaining at the end of the prediction window. We therefore augment the tracking objective with a quadratic input penalty and a terminal quadratic cost. 

To write the cost as a quadratic program in the input, we stack the inputs and the predicted outputs over the horizon.
Define the stacked
control sequence
\begin{equation}
U_i^k
:=
[
u_i(k)^\top \,
u_i(k+1)^\top \,
\cdots \,
u_i(k+T-1)^\top
]^{\top}
\in\mathbb R^{mT}.
\end{equation}
The predicted output satisfies
\begin{equation}
Y_i^k
=
\Theta_i U_i^k+\Phi_i x_i(k),
\label{eq:Y_i_affine}
\end{equation}
where $\Theta_i$ is the block lower-triangular matrix whose $(h,l)$ block is $C_iA_i^{h-l}B_i$ for $h\ge l$ and $\mathbf 0$ otherwise, and $\Phi_i:=\begin{bmatrix}(C_iA_i)^{\top}&(C_iA_i^2)^{\top}&\cdots&(C_iA_i^T)^{\top}\end{bmatrix}^{\top}$.
Moreover, the terminal state satisfies
\begin{equation}
x_i(k+T) = A_i^T x_i(k) + \mathcal B_i U_i^k,
\label{eq:terminal_state_lifted}
\end{equation}
where $
\mathcal B_i := \begin{bmatrix} A_i^{T-1}B_i & A_i^{T-2}B_i & \cdots & B_i \end{bmatrix}.
$

An agent that holds its output at the barycenter may need a nonzero input and state. Penalizing the absolute values would pull the agent away from the barycenter. Therefore, both penalties are measured as deviations from the input and the state that hold the output at the barycenter.

Such a state and input exist only under a condition on the agent. The following assumption is the standard solvability condition for offset-free tracking of constant outputs, and it covers general LTI agents.

\begin{assumption}[Standing Assumptions]
\label{ass:standing}
$(A_i,B_i)$ is controllable, $C_i$ has full row rank, and
\begin{equation}
\begin{bmatrix}A_i-\mathbf I_n & B_i\\ C_i&\mathbf 0\end{bmatrix}\in\R^{(n+d)\times(n+m)}\quad\text{has full row rank } n+d .
\label{eq:rank_cond}
\end{equation}
Equivalently, there exist $X_{e,i}\in\R^{n\times d}$ and $U_{e,i}\in\R^{m\times d}$ such that
\begin{equation}
A_iX_{e,i}+B_iU_{e,i}=X_{e,i},\qquad C_iX_{e,i}=\mathbf I_d .
\label{eq:hold}
\end{equation}
\end{assumption}
Condition \eqref{eq:rank_cond} requires $m\ge d$, ensuring that any constant output $y\equiv\bar q$ can be sustained by a constant input, and places no restriction on the eigenvalues of $A_i$. If this condition fails, no controller can hold a constant output exactly, making it a system property rather than a limitation of our method. When $m>d$, any fixed choice of $(X_{e,i},U_{e,i})$, such as the minimum-norm solution, may be utilized. For instance, single and double integrators trivially satisfy this with $U_{e,i}=\mathbf 0$. Consequently, for the current barycenter $\bar q_i(k)$, the pair $(X_{e,i}\bar q_i,\,U_{e,i}\bar q_i)$ serves as an equilibrium of \eqref{eq:dyn} that yields the desired output $\bar q_i$.

Applying \eqref{eq:Y_i_affine} and \eqref{eq:terminal_state_lifted} to this equilibrium trajectory gives, for every $q\in\R^d$,
\begin{equation}
\begin{aligned}
\Phi_iX_{e,i}q+\Theta_i(\one_T\otimes U_{e,i}q)=\one_T\otimes q,\\
A_i^TX_{e,i}q+\mathcal B_i(\one_T\otimes U_{e,i}q)=X_{e,i}q .
\label{eq:equil}
\end{aligned}
\end{equation}
Define the state deviation, and the input deviation stacked over the horizon, as
\begin{equation}
\tilde x_i(k):=x_i(k)-X_{e,i}\bar q_i(k),\quad
\tilde U_i^k:=U_i^k-\one_T\otimes U_{e,i}\bar q_i(k),
\label{eq:xtilde}
\end{equation}
so that $y_i(k)-\bar q_i(k)=C_i\tilde x_i(k)$. Subtracting \eqref{eq:equil} (with $q=\bar q_i(k)$) from \eqref{eq:Y_i_affine} and \eqref{eq:terminal_state_lifted} gives
\begin{equation}
\begin{aligned}
Y_i^k-\one_T\otimes\bar q_i(k)&=\Theta_i\tilde U_i^k+\Phi_i\tilde x_i(k),\\
x_i(k{+}T)-X_{e,i}\bar q_i(k)&=A_i^T\tilde x_i(k)+\mathcal B_i\tilde U_i^k .
\end{aligned}
\label{eq:Y_tilde}
\end{equation}
These coordinates are centered at the current barycenter $\bar q_i(k)$. Consequently, they describe the stationary reference.

Now, we introduce the terminal cost
$
V_{f,i}(\tilde x_i) = \tilde x_i^\top P_i \tilde x_i,
$
which penalizes the deviation $\tilde x_i\in\R^n$ of the terminal state from the equilibrium state $X_{e,i}\bar q_i(k)$ that holds the current barycenter, as defined in \eqref{eq:xtilde}, where $P_i\succeq0$ is a terminal weight. A natural choice is the stabilizing solution of the discrete-time algebraic Riccati equation associated with the auxiliary LQ regulator
\begin{equation}
P_i = A_i^\top P_iA_i - A_i^\top P_iB_i (B_i^\top P_iB_i+R_i)^{-1} B_i^\top P_iA_i + Q_i,
\label{eq:DARE}
\end{equation}
with $Q_i \succ 0$ and $R_i \succ 0$. Since $(A_i,B_i)$ is controllable (hence stabilizable) and $Q_i\succ0$, the stabilizing Riccati solution induces the terminal feedback
\begin{equation}
K_i := (B_i^\top P_iB_i+R_i)^{-1} B_i^\top P_iA_i,
\label{eq:lq_gain}
\end{equation}
for which $A_{K,i} := A_i - B_i K_i$ is Schur stable. 

The terminal penalty does not replace the Wasserstein tracking objective. It is a tail cost connecting the finite-horizon problem with an infinite-horizon LQ policy. It does not by itself certify closed-loop stability of the receding-horizon loop, because the stage cost has no state weight $Q_i$. Closed-loop behavior is treated in Section~\ref{sec:analysis}.

The stationary objective of the original D$^2$OC, in which the reference is frozen at the current barycenter over the whole horizon, is
\begin{equation}
\begin{aligned}
\mathcal{J}_i^{\mathrm{LQ}}(U)
&=\left\|\boldsymbol\Omega_i\big(\Theta_iU+\Phi_ix_i(k)-\one_T\otimes\bar q_i(k)\big)\right\|^2\\
&\quad+(U-\one_T\otimes U_{e,i}\bar q_i)^\top\mathcal R_i(U-\one_T\otimes U_{e,i}\bar q_i)\\
&\quad+V_{f,i}\big(x_i(k{+}T)-X_{e,i}\bar q_i(k)\big),
\label{eq:J^LQ}
\end{aligned}
\end{equation}
where $\mathcal{R}_i := \mathbf I_T \otimes R_i \succ 0$. By \eqref{eq:xtilde} and \eqref{eq:Y_tilde}, with $\tilde U=U-\one_T\otimes U_{e,i}\bar q_i$, it is equivalent to
\begin{equation*}
\mathcal{J}_i^{\mathrm{LQ}}=\left\|\boldsymbol\Omega_i(\Theta_i\tilde U+\Phi_i\tilde x_i)\right\|^2+\tilde U^\top\mathcal R_i\tilde U+V_{f,i}\big(A_i^T\tilde x_i+\mathcal B_i\tilde U\big).
\end{equation*}
\begin{align}
\text{Let}\quad
\mathcal M_i&:=\boldsymbol\Omega_i\Theta_i\label{eq:M_i}\\ 
\mathcal H_i^{\mathrm{LQ}}&:=2\left(\mathcal M_i^\top\mathcal M_i+\mathcal R_i+\mathcal B_i^\top P_i\mathcal B_i\right),\label{eq:LQ_Hessian}\\
\bar K_i&:=2(\mathcal H_i^{\mathrm{LQ}})^{-1}\big(\mathcal M_i^\top\boldsymbol\Omega_i\Phi_i+\mathcal B_i^\top P_iA_i^T\big).\label{eq:Kbar}
\end{align}

Then, the following result holds.

\begin{proposition}[Stationary LQ-D$^2$OC Solution]
\label{prop:stationary_lq}
Suppose $\mathcal R_i\succ0$ and $P_i\succeq0$. Then, $\mathcal H_i^{\mathrm{LQ}}\succ0$, $\mathcal J_i^{\mathrm{LQ}}$ is a strictly convex quadratic function of $U$, and its unique minimizer is
\begin{equation}
(U_i^k)^*=\one_T\otimes U_{e,i}\bar q_i(k)-\bar K_i\tilde x_i(k).
\label{eq:stationary_solution}
\end{equation}
\end{proposition}

\begin{proof}
Let $r_i:=\boldsymbol\Omega_i\Phi_i\tilde x_i(k)$, so that $\boldsymbol\Omega_i(\Theta_i\tilde U+\Phi_i\tilde x_i)=\mathcal M_i\tilde U+r_i$. Expanding the three terms of the equivalent form of $\mathcal J_i^{\mathrm{LQ}}$ gives
$\mathcal J_i^{\mathrm{LQ}}=\tfrac12\tilde U^\top\mathcal H_i^{\mathrm{LQ}}\tilde U+(f_i^{\mathrm{LQ}})^\top\tilde U+c_i$, where $f_i^{\mathrm{LQ}}=2\mathcal M_i^\top r_i+2\mathcal B_i^\top P_iA_i^T\tilde x_i(k)=\mathcal H_i^{\mathrm{LQ}}\bar K_i\tilde x_i(k)$ by the definition of $\bar K_i$, and $c_i$ does not depend on $\tilde U$. Because $\mathcal R_i\succ0$ and $P_i\succeq0$, $\mathcal H_i^{\mathrm{LQ}}\succeq2\mathcal R_i\succ0$, and thus $\mathcal J_i^{\mathrm{LQ}}$ is strictly convex with the unique minimizer $\tilde U^*=-(\mathcal H_i^{\mathrm{LQ}})^{-1}f_i^{\mathrm{LQ}}=-\bar K_i\tilde x_i(k)$. Returning to $U=\tilde U+\one_T\otimes U_{e,i}\bar q_i(k)$ gives \eqref{eq:stationary_solution}.
\end{proof}

This is the stationary D$^2$OC solution with LQ regularization. For integrator-type agents ($U_{e,i}=\mathbf 0$) it reduces to $-\bar K_i\tilde x_i(k)$.

\section{Predictive Compensation for Density-Flow Induced Lag}
\label{sec:feedforward}

We next incorporate the predicted motion of the local barycenter into
the LQ-regularized D$^2$OC objective.

\subsection{Feedforward-Augmented LQ-D$^2$OC}

Because the local sample selection is fixed over the horizon, the predicted barycenter $\bar q_i(k+h|k)$ in \eqref{eq:pred_barycenter} is the weighted average of the advected samples. For the first step, this gives
{\small
\begin{equation*}
\Delta\bar q_{i,1}(k)
:=
\bar q_i(k+1|k)-\bar q_i(k)
=
\Delta t
\frac{
\sum_{j\in\mathcal S_i(k)}
\pi_j(k)v(q_j(k),k)
}{
\sum_{j\in\mathcal S_i(k)}
\pi_j(k)
}.
\end{equation*}
}
For the whole prediction horizon, define the predicted drift relative to the current barycenter
\begin{equation}
\begin{aligned}
\Delta\bar q_{i,h}(k)&:=\bar q_i(k+h|k)-\bar q_i(k).
\label{eq:drift}
\end{aligned}
\end{equation}
Since the barycenter is a convex average, Assumption~\ref{ass:vbound} leads to $\|\Delta\bar q_{i,h}(k)\|\le h\Delta t\,\bar v$, hence, with $
\Delta\bar Q_i^k
:=
\begin{bmatrix}
\Delta\bar q_{i,1}(k)^\top &
\cdots &
\Delta\bar q_{i,T}(k)^\top
\end{bmatrix}^{\top}\in\R^{dT},
$
\begin{equation}
\|\Delta\bar Q_i^k\|^2\le(\Delta t\,\bar v)^2\frac{T(T+1)(2T+1)}{6}.
\label{eq:drift_bound}
\end{equation}

Define the augmented barycenter trajectory over the horizon as
\begin{equation}
\bar Q_i^{\xi}
:=
\one_T\otimes\bar q_i(k)+\xi\Delta\bar Q_i^k,
\qquad
\xi\in[0,1].
\label{eq:aug_reference}
\end{equation}
For $\xi=0$ the reference is frozen at the current barycenter, which is the stationary D$^2$OC reference. For $\xi=1$, $\bar Q_i^{k}=[\bar q_i(k+1|k)^\top\cdots\bar q_i(k+T|k)^\top]^\top$ is the advected reference of Proposition~\ref{proposition:equiv}. Thus, tracking it minimizes the predicted local Wasserstein cost. 

With $\mathcal U_{e,i}:=\mathbf I_T\otimes U_{e,i}$ and $\mathbf J_T:=[\mathbf 0\ \cdots\ \mathbf 0\ \ \mathbf I_d]\in\R^{d\times dT}$, so that $\mathbf J_T\bar Q_i^{\xi}=\bar q_i(k)+\xi\Delta\bar q_{i,T}(k)$, consider the feedforward D$^2$OC objective
\begin{equation}
\begin{split}
&\mathcal J_i^{\mathrm{ff}}(U)
:= \big\|\boldsymbol{\Omega}_i\big(\Theta_i U+\Phi_ix_i(k) -\bar Q_i^{\xi}\big)\big\|^2 
+ (U-\mathcal U_{e,i}\bar Q_i^{\xi})^\top \\
&\quad\times\mathcal R_i(U-\mathcal U_{e,i}\bar Q_i^{\xi}) 
+ \big\|A_i^Tx_i(k)+\mathcal B_iU-X_{e,i}\mathbf J_T\bar Q_i^{\xi}\big\|_{P_i}^2,
\end{split}
\label{eq:J_ff_LQ}
\end{equation}
where $\mathcal{R}_i\succ0$ and $P_i\succeq0$ are the input and terminal
weighting matrices, respectively. The first term corresponds to $\|\boldsymbol\Omega_i(Y_i^k-\bar Q_i^\xi)\|^2$, the second penalizes the deviation from the equilibrium input $U_{e,i}\bar q_i^\xi(k{+}h)$ of each reference position, with $\bar q_i^\xi(k{+}h):=\bar q_i(k)+\xi\Delta\bar q_{i,h}(k)$ the $h$-th block of $\bar Q_i^\xi$, and the third equals $\|x_i(k{+}T)-X_{e,i}\bar q_i^{\xi}(k{+}T)\|_{P_i}^2$, the deviation from the reference equilibrium state at the end of the horizon. Note that for $\xi=0$, \eqref{eq:J_ff_LQ} equals $\mathcal J_i^{\mathrm{LQ}}$ in \eqref{eq:J^LQ}.

For notational convenience, define 
\begin{align}
G_i^{\mathrm{ff}}&:=2(\mathcal H_i^{\mathrm{LQ}})^{-1}\big(\mathcal M_i^\top\boldsymbol\Omega_i+\mathcal R_i\mathcal U_{e,i}+\mathcal B_i^\top P_iX_{e,i}\mathbf J_T\big),
\label{eq:ff_gain}
\end{align}
where $\mathcal M_i^\top\boldsymbol\Omega_i=\alpha_i\Theta_i^\top$ with $\alpha_i=\sum_{j\in\mathcal S_i(k)}\pi_j(k)$. For $U_{e,i}=\mathbf 0$ the term $\mathcal R_i\mathcal U_{e,i}$ vanishes.

We next establish the explicit optimal control law.

\begin{theorem}[Feedforward-Augmented LQ-D$^2$OC Control Law]
\label{thm:ff_controller}
Consider the feedforward D$^2$OC problem
\eqref{eq:J_ff_LQ} with the augmented reference
\eqref{eq:aug_reference}. Suppose that $\mathcal{R}_i\succ0$ and
$P_i\succeq0$, and that Assumption~\ref{ass:standing} holds. Then, \eqref{eq:J_ff_LQ} is a strictly convex quadratic
function of $U$ and admits a unique global minimizer
\begin{equation}
U_i^{\mathrm{ff}}
=
G_i^{\mathrm{ff}}\bar Q_i^{\xi}-\bar K_i\,x_i(k).
\label{eq:proposed_ff_controller}
\end{equation}
Equivalently, the optimal sequence decomposes as
\begin{equation}
U_i^{\mathrm{ff}}
=
U_i^{\mathrm{reactive}}
+
U_i^{\mathrm{comp}},
\label{eq:ff_controller_decomposition}
\end{equation}
\begin{flalign}
&\text{where}&
U_i^{\mathrm{reactive}}
 &:=
\one_T\otimes U_{e,i}\bar q_i(k)-\bar K_i\tilde x_i(k),&
\label{eq:reactive_lq_controller}
\\
& &
U_i^{\mathrm{comp}}
&:=
\xi\,G_i^{\mathrm{ff}}\,\Delta\bar Q_i^k.&
\label{eq:feedforward_compensation}
\end{flalign}
The second term $U_i^{\mathrm{comp}}$ in $U_i^{\mathrm{ff}}$ is the feedforward compensation induced by the predicted barycenter drift, and the reactive part is the stationary D$^2$OC solution, so that $U_i^{\mathrm{ff}}=U_i^{\mathrm{ff}}|_{\xi=0}+\xi G_i^{\mathrm{ff}}\Delta\bar Q_i^k$.
\end{theorem}

\begin{proof}
Let $x = x_i(k)$. Expanding the three terms of \eqref{eq:J_ff_LQ} and collecting powers of $U$ yields
\begin{equation}
\mathcal J_i^{\mathrm{ff}}(U) = \tfrac12 U^\top \mathcal H_i^{\mathrm{LQ}} U + (f^\xi)^\top U + \mathrm{const.},
\end{equation}
where, using \eqref{eq:Kbar} and \eqref{eq:ff_gain},
$f^\xi=2\mathcal M_i^\top\boldsymbol\Omega_i(\Phi_ix-\bar Q_i^\xi)-2\mathcal R_i\mathcal U_{e,i}\bar Q_i^\xi+2\mathcal B_i^\top P_i(A_i^Tx-X_{e,i}\mathbf J_T\bar Q_i^\xi)=\mathcal H_i^{\mathrm{LQ}}\bar K_ix-\mathcal H_i^{\mathrm{LQ}}G_i^{\mathrm{ff}}\bar Q_i^\xi$.
Since $\mathcal R_i\succ0$ and $P_i\succeq0$, $\mathcal H_i^{\mathrm{LQ}}\succeq2\mathcal R_i\succ0$, $\mathcal J_i^{\mathrm{ff}}$ is strictly convex and its unique minimizer is $-(\mathcal H_i^{\mathrm{LQ}})^{-1}f^\xi=G_i^{\mathrm{ff}}\bar Q_i^\xi-\bar K_ix$, which establishes \eqref{eq:proposed_ff_controller}.

For the decomposition, fix any $q\in\R^d$ and take $x=X_{e,i}q$ and $\bar Q_i^\xi=\one_T\otimes q$. By \eqref{eq:equil}, the sequence $U=\one_T\otimes U_{e,i}q$ makes all three terms of \eqref{eq:J_ff_LQ} vanish. The first because $\Phi_iX_{e,i}q+\Theta_iU=\one_T\otimes q$, the second because $\mathcal U_{e,i}(\one_T\otimes q)=\one_T\otimes U_{e,i}q$, and the third because $A_i^TX_{e,i}q+\mathcal B_iU=X_{e,i}q=X_{e,i}\mathbf J_T(\one_T\otimes q)$. The objective is nonnegative, and thus this $U$ is the unique minimizer, so that \eqref{eq:proposed_ff_controller} gives $G_i^{\mathrm{ff}}(\one_T\otimes q)-\bar K_iX_{e,i}q=\one_T\otimes U_{e,i}q$. As $q$ is arbitrary,
\begin{equation}
G_i^{\mathrm{ff}}(\one_T\otimes\mathbf I_d)-\bar K_iX_{e,i}=\one_T\otimes U_{e,i}.
\label{eq:equil_gain}
\end{equation}
Now, substitute the definition $\bar Q_i^\xi=\one_T\otimes\bar q_i(k)+\xi\Delta\bar Q_i^k$ into \eqref{eq:proposed_ff_controller} and apply \eqref{eq:equil_gain} with $q=\bar q_i(k)$:
\begin{equation*}
\begin{aligned}
U_i^{\mathrm{ff}}&=G_i^{\mathrm{ff}}(\one_T\otimes\bar q_i(k)+\xi\Delta\bar Q_i^k)-\bar K_ix\\
&=\big(G_i^{\mathrm{ff}}(\one_T\otimes\mathbf I_d)\bar q_i(k)-\bar K_ix\big)+\xi G_i^{\mathrm{ff}}\Delta\bar Q_i^k\\
&=\big(\bar K_iX_{e,i}\bar q_i(k)+\one_T\otimes U_{e,i}\bar q_i(k)-\bar K_ix\big)+\xi G_i^{\mathrm{ff}}\Delta\bar Q_i^k\\
&=\one_T\otimes U_{e,i}\bar q_i(k)-\bar K_i\tilde x_i(k)+\xi G_i^{\mathrm{ff}}\Delta\bar Q_i^k,
\end{aligned}
\end{equation*}
which directly yields \eqref{eq:ff_controller_decomposition}--\eqref{eq:feedforward_compensation}.
\end{proof}

\begin{remark}[Unconstrained Problem]
Theorem~\ref{thm:ff_controller} is the closed-form solution of the unconstrained problem. With active input constraints the minimizer is no longer linear in $\Delta\bar Q_i^k$, which is not analyzed here.
\end{remark}

\subsection{Lag Reduction via Feedforward-Augmented LQ-D$^2$OC}

The terminal Riccati construction provides an auxiliary LQ regulator.
Let $P_i$ be the stabilizing
solution of the discrete-time algebraic Riccati equation \eqref{eq:DARE}
with $Q_i\succ0$ and $R_i\succ0$, and let $K_i$ be given by
\eqref{eq:lq_gain}. Then,
$A_{K,i}:=A_i-B_iK_i$
is Schur stable. Moreover, the terminal value function
$V_{f,i}(\tilde x_i):=\tilde x_i^\top P_i \tilde x_i$
satisfies
\begin{equation*}
V_{f,i}(A_{K,i}\tilde x_i)-V_{f,i}(\tilde x_i)
=
-\tilde x_i^\top
\left(
Q_i+K_i^\top R_iK_i
\right)\tilde x_i.
\end{equation*}
Thus, the auxiliary LQ regulator has a strict Lyapunov decrease. This identity concerns the auxiliary regulator, not the receding-horizon loop itself, whose stability is treated later.

First, we quantify the lag. Define the vector
\begin{equation}
g_i:=\left(\mathcal M_i^\top\boldsymbol\Omega_i+\mathcal R_i\mathcal U_{e,i}+\mathcal B_i^\top P_iX_{e,i}\mathbf J_T\right)\Delta\bar Q_i^k\in\R^{mT},
\label{eq:g_def}
\end{equation}
so that $G_i^{\mathrm{ff}}\Delta\bar Q_i^k=2(\mathcal H_i^{\mathrm{LQ}})^{-1}g_i$, and let
\begin{equation}
\begin{aligned}
\Gamma_i:=C_iB_i\,E_1G_i^{\mathrm{ff}}\in\R^{d\times dT},\;
E_1:=[\mathbf I_m\ \mathbf 0\ \cdots\ \mathbf 0]\in\mathbb{R}^{m\times mT}.
\label{eq:Gamma_ff}
\end{aligned}
\end{equation}
To analyze the tracking performance, let $e_i(k) := y_i(k) - \bar q_i(k)$ denote the current tracking error. We define the predicted next reference barycenter as $\hat q_i(k+1) := \bar q_i(k) + \Delta \bar q_{i,1}(k)$, which is independent of the compensation factor $\xi$. Under the feedforward control law $U_i^{\mathrm{ff}}$, the predicted one-step-ahead output is given by
$
y_i^\xi(k+1) = C_i A_i x_i(k) + C_i B_i E_1 U_i^{\mathrm{ff}}.
$
Accordingly, the one-step-ahead tracking error under the $\xi$-controller is defined as
$
e_i^\xi(k+1) := y_i^\xi(k+1) - \hat q_i(k+1).
$
The actual reference barycenter at the next step deviates from the predicted trajectory by a perturbation term $\varpi_i(k)$ introduced in Remark \ref{remark:discretization_error}, such that $\bar q_i(k+1) = \hat q_i(k+1) + \varpi_i(k)$. Note that $\varpi_i(k)$ may implicitly depend on $\xi$ through the updated selection of the local reference set $\mathcal S_i(k+1)$. 

Evaluating the feedforward objective at full compensation ($\xi=1$) gives the true predicted cost $\mathcal J_i^{\mathrm{ff}}|_{\xi=1}$, whose tracking component equals the predicted local Wasserstein cost up to an additive constant by Proposition~\ref{proposition:equiv}.

The following theorem formalizes the lag-reduction property of the
proposed feedforward augmentation.

\begin{theorem}[Exact Lag Quantification]
\label{thm:lag_reduction}
Consider the feedforward-augmented LQ-D$^2$OC controller in
Theorem~\ref{thm:ff_controller}. Write $U_i^{\mathrm{ff}}|_{\xi}$ for its minimizer \eqref{eq:proposed_ff_controller} at a given $\xi\in[0,1]$, and $\mathcal J_i^{\star}:=\mathcal J_i^{\mathrm{ff}}|_{\xi=1}$ for the true objective, whose minimizer is $U_i^{\mathrm{ff}}|_{\xi=1}$. For every $\xi\in[0,1]$:
\begin{enumerate}[label=(\roman*), leftmargin=1.5em, labelsep=0.5em]
\item (Horizon cost) Applying $U_i^{\mathrm{ff}}|_{\xi}$ to the true objective incurs the excess cost
\begin{equation}
\begin{aligned}
\mathcal J_i^{\star}\big(U_i^{\mathrm{ff}}|_{\xi}\big)-\mathcal J_i^{\star}\big(U_i^{\mathrm{ff}}|_{\xi=1}\big)
&=
2(1-\xi)^2\,g_i^\top(\mathcal H_i^{\mathrm{LQ}})^{-1}g_i\ge0.
\label{eq:excess_cost}
\end{aligned}
\end{equation}
Equality holds iff $\xi=1$ or $g_i=0$. Moreover,
\begin{equation}
\begin{aligned}
2(1-\xi)^2g_i^\top(\mathcal H_i^{\mathrm{LQ}})^{-1}g_i&\le\frac{(1-\xi)^2\|g_i\|^2}{\lambda_{\min}(R_i)},\\
\|g_i\|&\le c_{g,i}\|\Delta\bar Q_i^k\|,
\label{eq:excess_bound}
\end{aligned}
\end{equation}
with $c_{g,i}:=\alpha_i\|\Theta_i\|+\|R_iU_{e,i}\|+\|\mathcal B_i\|\|P_i\|\|X_{e,i}\|$ and $\|\Delta\bar Q_i^k\|$ bounded by \eqref{eq:drift_bound}.
\item (First-step error) Let $e_i^{\xi}(k{+}1)$ be the first-step error with respect to $\hat q_i(k{+}1)$ when the first input of $U_i^{\mathrm{ff}}|_{\xi}$ is applied. Then,
\begin{equation}
\begin{aligned}
e_i^{\xi}(k{+}1)
&=
e_i^{1}(k{+}1)-(1-\xi)\,\Gamma_i\,\Delta\bar Q_i^k\\
&=
(1-\xi)\,e_i^{0}(k{+}1)+\xi\,e_i^{1}(k{+}1).
\label{eq:lag_reduced_error}
\end{aligned}
\end{equation}
Consequently,
\begin{equation}
\|e_i^\xi(k{+}1)\|\le(1-\xi)\|e_i^0(k{+}1)\|+\xi\|e_i^1(k{+}1)\|.
\label{eq:lag_reduction_bound}
\end{equation}
\end{enumerate}
\end{theorem}

\begin{proof}
\textit{(i)} By Theorem~\ref{thm:ff_controller}, the two minimizers differ by
\begin{equation*}
U_i^{\mathrm{ff}}|_{\xi}-U_i^{\mathrm{ff}}|_{\xi=1}
=-(1-\xi)\,G_i^{\mathrm{ff}}\Delta\bar Q_i^k
=-2(1-\xi)(\mathcal H_i^{\mathrm{LQ}})^{-1}g_i .
\end{equation*}
The objective $\mathcal J_i^{\star}$ is quadratic in $U$ with Hessian $\mathcal H_i^{\mathrm{LQ}}$ and minimizer $U_i^{\mathrm{ff}}|_{\xi=1}$, hence
\begin{equation*}
\mathcal J_i^{\star}(U)-\mathcal J_i^{\star}\big(U_i^{\mathrm{ff}}|_{\xi=1}\big)
=\tfrac12\big(U-U_i^{\mathrm{ff}}|_{\xi=1}\big)^\top\mathcal H_i^{\mathrm{LQ}}\big(U-U_i^{\mathrm{ff}}|_{\xi=1}\big).
\end{equation*}
Substituting $U=U_i^{\mathrm{ff}}|_{\xi}$ gives
\begin{equation*}
\tfrac12\cdot4(1-\xi)^2\,g_i^\top(\mathcal H_i^{\mathrm{LQ}})^{-1}g_i
=2(1-\xi)^2\,g_i^\top(\mathcal H_i^{\mathrm{LQ}})^{-1}g_i,
\end{equation*}
which is \eqref{eq:excess_cost}. 

For the bounds, $\mathcal H_i^{\mathrm{LQ}}\succeq2\mathcal R_i\succeq2\lambda_{\min}(R_i)\mathbf I$ gives $g_i^\top(\mathcal H_i^{\mathrm{LQ}})^{-1}g_i\le\|g_i\|^2/(2\lambda_{\min}(R_i))$, which is the first inequality in \eqref{eq:excess_bound}. The second follows from the triangle inequality applied to \eqref{eq:g_def}, with $\|\mathbf J_T\|=1$ and $\|\mathcal M_i^\top\boldsymbol\Omega_i\|=\alpha_i\|\Theta_i\|$. Finally, since $(\mathcal H_i^{\mathrm{LQ}})^{-1}\succ0$, the excess cost is zero if and only if $(1-\xi)g_i=0$.

\textit{(ii)} Since $x_i(k)$ and $\hat q_i(k{+}1)$ do not depend on $\xi$, subtracting $e_i^1(k{+}1)=y_i^1(k{+}1)-\hat q_i(k{+}1)$ from $e_i^\xi(k{+}1)$ cancels $C_iA_ix_i(k)$ and gives
\begin{equation*}
\begin{aligned}
e_i^\xi(k{+}1)-e_i^1(k{+}1)
&=C_iB_iE_1\big(U_i^{\mathrm{ff}}|_{\xi}-U_i^{\mathrm{ff}}|_{\xi=1}\big)\\
&=-(1-\xi)\,\Gamma_i\Delta\bar Q_i^k,
\end{aligned}
\end{equation*}
the first equality in \eqref{eq:lag_reduced_error}. Setting $\xi=0$ gives $e_i^0=e_i^1-\Gamma_i\Delta\bar Q_i^k$, hence $(1-\xi)e_i^0+\xi e_i^1=e_i^1-(1-\xi)\Gamma_i\Delta\bar Q_i^k$, the second equality. Taking norms and applying the triangle inequality gives \eqref{eq:lag_reduction_bound}.
\end{proof}

\begin{remark}[Interpretation of Theorem~\ref{thm:lag_reduction}]\label{remark:lag_reduction}
Part~(i) is a statement over the \emph{entire prediction horizon}. By \eqref{eq:excess_bound}, the lag cost that the stationary protocol ($\xi=0$) accumulates relative to the advected reference grows at most quadratically with the predicted drift $\|\Delta\bar Q_i^k\|$, decreases as $(1-\xi)^2$ when $\xi$ increases, and vanishes at $\xi=1$. Hence, $\xi=1$ is optimal with respect to the predicted horizon cost, provided that the predicted barycenters are accurate.

Part~(ii) is a statement about the \emph{first step} only, which is the step actually applied in the receding-horizon loop. The first-step error is affine in $\xi$: $e_i^{\xi}(k{+}1)=e_i^{0}(k{+}1)+\xi\,\Gamma_i\Delta\bar Q_i^k$. Whether this shift reduces the error depends on the direction of $\Gamma_i\Delta\bar Q_i^k$ relative to $e_i^{0}(k{+}1)$, and the horizon-optimal plan may accept a larger first-step error in exchange for a smaller cost at later steps. Thus, the theorem does not assert $\|e_i^1\|<\|e_i^0\|$, and the value of $\xi$ minimizing the realized error depends on the agent dynamics and on the prediction accuracy. We therefore treat $\xi$ as a design parameter between purely reactive tracking ($\xi=0$) and full anticipation ($\xi=1$). If $C_iB_i=\mathbf 0$, then $\Gamma_i=\mathbf 0$ and the first-step error does not depend on $\xi$.
\end{remark}

\section{Analysis on Tracking Error and Wasserstein Bound}\label{sec:analysis}
The previous lag-quantification result concerns the transient tracking
mechanism induced by reference motion. We next establish a separate
closed-loop result showing that the remaining prediction
and implementation errors do not destroy boundedness.

\subsection{Closed-Loop Boundedness Analysis}
\begin{lemma}[Receding-Horizon Input]
\label{lem:rh_input}
Let $E_1:=[\mathbf I_m\ \mathbf 0\ \cdots\ \mathbf 0]\in\R^{m\times mT}$ select the first input block, and let $K_i^{\mathrm{mpc}}:=E_1\bar K_i$. Under the assumptions of Theorem~\ref{thm:ff_controller}, the unconstrained receding-horizon input $u_i(k):=E_1U_i^{\mathrm{ff}}$ is
\begin{equation}
u_i(k)=U_{e,i}\bar q_i(k)-K_i^{\mathrm{mpc}}\tilde x_i(k)+\xi E_1G_i^{\mathrm{ff}}\Delta\bar Q_i^k .
\label{eq:closed_loop_input}
\end{equation}
\end{lemma}
\begin{proof}
Apply $E_1$ to the decomposition \eqref{eq:ff_controller_decomposition} with \eqref{eq:reactive_lq_controller} and \eqref{eq:feedforward_compensation}. Since $E_1(\one_T\otimes U_{e,i})=U_{e,i}$ and $E_1\bar K_i=K_i^{\mathrm{mpc}}$, the first block of $U_i^{\mathrm{reactive}}$ is $U_{e,i}\bar q_i(k)-K_i^{\mathrm{mpc}}\tilde x_i(k)$, and the first block of $U_i^{\mathrm{comp}}$ is $\xi E_1G_i^{\mathrm{ff}}\Delta\bar Q_i^k$.
\end{proof}

According to Remark \ref{remark:discretization_error}, the actual next barycenter, which differs from the predicted one, can be updated by including the total perturbation $\varpi_i(k)$ as follows:
\begin{equation}
\bar q_i(k+1)=\bar q_i(k)+\Delta\bar q_{i,1}(k)+\varpi_i(k).
\label{eq:barycenter_perturbed}
\end{equation}

Define $\tilde u_i(k):=u_i(k)-U_{e,i}\bar q_i(k)=-K_i^{\mathrm{mpc}}\tilde x_i(k)+\xi E_1G_i^{\mathrm{ff}}\Delta\bar Q_i^k$. 
Using \eqref{eq:barycenter_perturbed}, the equilibrium relationship $A_iX_{e,i}+B_iU_{e,i}=X_{e,i}$ of \eqref{eq:hold}, and the definition \eqref{eq:xtilde}, we have
\begin{equation*}
\begin{aligned}
\tilde x_i(k+1) &= A_ix_i+B_iu_i-X_{e,i}\bar q_i(k+1)\\
&= A_i\tilde x_i+B_i\tilde u_i+(A_iX_{e,i}+B_iU_{e,i})\bar q_i(k) \\
&\quad - X_{e,i}\bar q_i(k+1)\\
&= A_i\tilde x_i+B_i\tilde u_i-X_{e,i}\big(\Delta\bar q_{i,1}(k)+\varpi_i(k)\big),
\end{aligned}
\end{equation*}
leading to
\begin{equation}
\tilde x_i(k+1)
=
A_{\mathrm{cl},i}\tilde x_i(k)
+
d_i(k),\quad A_{\mathrm{cl},i}:=A_i-B_iK_i^{\mathrm{mpc}},
\label{eq:deviation_closed_loop_compact}
\end{equation}
\begin{equation}
d_i(k)
:=
-\Big(X_{e,i}\Delta\bar q_{i,1}(k)-\xi B_iE_1G_i^{\mathrm{ff}}\Delta\bar Q_i^k\Big)-X_{e,i}\varpi_i(k).
\label{eq:total_disturbance}
\end{equation}
The disturbance $d_i(k)$ consists of the drift term $X_{e,i}\Delta\bar q_{i,1}(k)$, the part $\xi B_iE_1G_i^{\mathrm{ff}}\Delta\bar Q_i^k$ contributed by the feedforward, and the perturbation $X_{e,i}\varpi_i(k)$.

\begin{assumption}[Nominal Stability]
\label{ass:nominal}
The matrix $A_{\mathrm{cl},i}=A_i-B_iK_i^{\mathrm{mpc}}$ is Schur stable.
\end{assumption}
Assumption~\ref{ass:nominal} states that the closed loop of the stationary D$^2$OC controller, without drift or perturbation, is stable. It is needed because, by \eqref{eq:deviation_closed_loop_compact}, $\tilde x_i$ is driven by the bounded disturbance $d_i(k)$ through $A_{\mathrm{cl},i}$. The matrix $A_{\mathrm{cl},i}$ depends on $(A_i,B_i,C_i)$, $\alpha_i$, $R_i$, $P_i$, and $T$, but not on $\xi$, because the feedforward enters only through $d_i(k)$. Therefore, the assumption is verified once, offline, by computing the eigenvalues of $A_{\mathrm{cl},i}$, and it holds for every $\xi\in[0,1]$. It does not require $A_i$ itself to be stable.

\begin{assumption}[Bounded Disturbances]
\label{ass:disturbance}
$\|\varpi_i(k)\|\le\zeta_i$ for all $k$ and all $\xi\in[0,1]$, and
\begin{equation}
\begin{aligned}
a_i &:= \sup_k\|X_{e,i}\Delta\bar q_{i,1}(k)\|, \\
b_i &:= \sup_k\big\|X_{e,i}\Delta\bar q_{i,1}(k)-B_iE_1G_i^{\mathrm{ff}}\Delta\bar Q_i^k\big\|
\end{aligned}
\end{equation}
are finite. By Assumption~\ref{ass:vbound}, $a_i\le\|X_{e,i}\|\Delta t\,\bar v$. Since $\varpi_i$ contains the reselection jitter, which depends on the closed-loop trajectory, $\zeta_i$ in Assumption~\ref{ass:disturbance} is a bound that holds along all closed-loop trajectories, for every $\xi$.
\end{assumption}

Since $A_{\mathrm{cl},i}$ is Schur stable, for every
$Q_{L,i}\succ0$, there exists a unique $S_i\succ0$ satisfying
$
A_{\mathrm{cl},i}^\top S_iA_{\mathrm{cl},i}-S_i
=
-Q_{L,i}.
$
The matrix $Q_{L,i}\succ0$ is a free design parameter of the analysis. For $0<\epsilon_i<\lambda_{\min}(Q_{L,i})$, define
\begin{equation}
\begin{aligned}
\sigma_i &:= \lambda_{\min}(Q_{L,i})-\epsilon_i, \quad
c_i' := \frac{\|A_{\mathrm{cl},i}^\top S_i\|^2}{\epsilon_i}+\|S_i\|, \\
\eta_i &:= \sqrt{\frac{c_i'\lambda_{\max}(S_i)}{\sigma_i\lambda_{\min}(S_i)}},\quad
\kappa_i := 1-\frac{\sigma_i}{\lambda_{\max}(S_i)},
\end{aligned}
\label{eq:rho_i}
\end{equation}
and the disturbance level
\begin{equation}
\bar d_i(\xi):=(1-\xi)\,a_i+\xi\,b_i+\|X_{e,i}\|\zeta_i .
\label{eq:dbar}
\end{equation}

Using these definitions, we now state the closed-loop boundedness result.

\begin{theorem}[Uniform Ultimate Boundedness of Feedforward D$^2$OC]
\label{thm:uub}
Suppose Assumptions~\ref{ass:vbound}, \ref{ass:standing}, \ref{ass:nominal}, and \ref{ass:disturbance} hold, $\mathcal R_i\succ0$ and $P_i\succeq0$. Then, $\|d_i(k)\|\le\bar d_i(\xi)$ for all $k$, and the tracking error generated by the feedforward D$^2$OC controller is uniformly
ultimately bounded. More precisely, for every $Q_{L,i}\succ0$ and every
$0<\epsilon_i<\lambda_{\min}(Q_{L,i})$,
\begin{equation}
\limsup_{k\rightarrow\infty}
\|\tilde x_i(k)\|
\leq
\eta_i\bar d_i(\xi),
\label{eq:x_ultimate_bound}
\end{equation}
where $V_i(k):=\tilde x_i(k)^\top S_i\tilde x_i(k)$ and, for all $k\ge0$,
\begin{equation}
\|\tilde x_i(k)\|^2\le\frac{1}{\lambda_{\min}(S_i)}\Big(\kappa_i^{k}V_i(0)+\frac{1-\kappa_i^{k}}{1-\kappa_i}\,c_i'\bar d_i(\xi)^2\Big),
\label{eq:x_transient}
\end{equation}
and consequently
\begin{equation*}
\limsup_{k\rightarrow\infty}
\|e_i(k)\|
\leq
\|C_i\|\eta_i\bar d_i(\xi).
\end{equation*}
\end{theorem}

\begin{proof}
By \eqref{eq:total_disturbance}, $d_i=-\vartheta_i-X_{e,i}\varpi_i$ with
\begin{equation*}
\begin{aligned}
\vartheta_i &:= X_{e,i}\Delta\bar q_{i,1}-\xi B_iE_1G_i^{\mathrm{ff}}\Delta\bar Q_i^k \\
&= (1-\xi)X_{e,i}\Delta\bar q_{i,1}+\xi\big(X_{e,i}\Delta\bar q_{i,1}-B_iE_1G_i^{\mathrm{ff}}\Delta\bar Q_i^k\big),
\end{aligned}
\end{equation*}
thus $\|d_i\|\le(1-\xi)a_i+\xi b_i+\|X_{e,i}\|\zeta_i=\bar d_i(\xi)$.

Using \eqref{eq:deviation_closed_loop_compact},
\begin{align*}
V_i(k+1)-V_i(k)
&=
-\tilde x_i(k)^\top Q_{L,i}\tilde x_i(k)
\nonumber\\
&+
2\tilde x_i(k)^\top
A_{\mathrm{cl},i}^\top S_i d_i(k)
+
d_i(k)^\top S_i d_i(k).
\end{align*}
For any $\epsilon_i>0$, Young's inequality gives
\begin{align*}
2
\left|
\tilde x_i^\top
A_{\mathrm{cl},i}^\top S_i d_i
\right|
\leq
\epsilon_i\|\tilde x_i\|^2
+
\frac{
\|A_{\mathrm{cl},i}^\top S_i\|^2
}{
\epsilon_i
}
\|d_i\|^2,
\end{align*}
and $d_i^\top S_id_i\le\|S_i\|\|d_i\|^2$. Consequently,
\begin{equation}
V_i(k+1)-V_i(k)
\leq
-\sigma_i
\|\tilde x_i(k)\|^2
+
c_i'\bar d_i(\xi)^2.
\label{eq:V_bound}
\end{equation}
Since
$V_i(k)\leq\lambda_{\max}(S_i)\|\tilde x_i(k)\|^2$,
we have
\begin{equation}
V_i(k+1)
\leq
\kappa_i V_i(k)
+
c_i'\bar d_i(\xi)^2.
\label{eq:V_scalar_recursion}
\end{equation}
From the Lyapunov equation, $S_i-A_{\mathrm{cl},i}^\top S_iA_{\mathrm{cl},i}=Q_{L,i}$, hence $S_i\succeq Q_{L,i}$ and
$\lambda_{\max}(S_i)\geq\lambda_{\min}(Q_{L,i})>\sigma_i>0$, which gives
$0<\kappa_i<1$.
Iterating \eqref{eq:V_scalar_recursion} yields
\begin{equation}
V_i(k)
\leq
\kappa_i^kV_i(0)
+
\frac{1-\kappa_i^k}{1-\kappa_i}
c_i'\bar d_i(\xi)^2.
\end{equation}
Since $V_i(k)\ge\lambda_{\min}(S_i)\|\tilde x_i(k)\|^2$, this gives \eqref{eq:x_transient}. Taking the limit superior, and using $\lambda_{\min}(S_i)\|\tilde x_i\|^2\le V_i$ and $1-\kappa_i=\sigma_i/\lambda_{\max}(S_i)$,
\begin{equation*}
\limsup_{k\rightarrow\infty}\|\tilde x_i(k)\|^2\le\frac{1}{\lambda_{\min}(S_i)}\frac{c_i'\bar d_i^2}{1-\kappa_i}=\frac{c_i'\lambda_{\max}(S_i)}{\sigma_i\lambda_{\min}(S_i)}\bar d_i^2,
\end{equation*}
which gives \eqref{eq:x_ultimate_bound}. Finally, since
$e_i(k)=C_i\tilde x_i(k)$, $\|e_i(k)\|\leq\|C_i\|\|\tilde x_i(k)\|$, and
\begin{equation*}
\limsup_{k\rightarrow\infty}
\|e_i(k)\|
\leq
\|C_i\|\eta_i\bar d_i(\xi).
\end{equation*}
This proves the uniform ultimate boundedness of the tracking error.
\end{proof}

\begin{remark}[Role of the Feedforward in the Bound]
The disturbance level $\bar d_i(\xi)$ is the convex combination $(1-\xi)a_i+\xi b_i$ plus the perturbation term. For a common $\zeta_i$, the feedforward lowers the bound exactly when $b_i<a_i$, and the reduction is linear in $\xi$:
\begin{equation*}
\bar d_i(0)-\bar d_i(\xi)=\xi\,(a_i-b_i).
\end{equation*}
Whether $b_i<a_i$ holds depends on the agent dynamics and the controller parameters. If $b_i\ge a_i$, the bound does not certify an improvement, which does not mean that the feedforward is ineffective. The bound treats the feedforward input as an additional disturbance and cannot capture its cancellation of the lag over the horizon. The ultimate bound is a sufficient guarantee of boundedness for every $\xi\in[0,1]$, not a tight characterization of the error.
\end{remark}

\subsection{Wasserstein Tracking Bound Analysis}
\label{subsec:W_bound}

We next translate the above state-space tracking bound into a corresponding
bound on the local Wasserstein tracking error. The key observation is
that the Wasserstein tracking error can be decomposed exactly into a
barycenter-tracking component and an irreducible spatial-spread
component.

For each agent $i$, the local weighted reference distribution $\mu_i(k):=\sum_{j\in\mathcal S_i(k)}\pi_j(k)\delta_{q_j(k)}$ has total mass $\alpha_i$ and weighted barycenter $\bar q_i(k)$ in \eqref{eq:local_barycenter_time_varying}.

The local quadratic Wasserstein tracking error is defined as
\begin{equation}
\mathcal W_i^2(k)
:=
\sum_{j\in\mathcal S_i(k)}
\pi_j(k)
\|y_i(k)-q_j(k)\|^2.
\label{eq:local_W_definition}
\end{equation}

We can now combine Lemma \ref{lem:W_barycenter_decomposition} with the closed-loop
ultimate bound established in Theorem~\ref{thm:uub}.

Assume that the local reference spread is uniformly bounded, i.e.,
there exists a finite constant $\bar{\psi}_i\geq0$ such that
\begin{equation}
\psi_i(k)
\leq
\bar{\psi}_i,
\qquad
\forall k\geq0.
\label{eq:uniform_reference_spread}
\end{equation}

Recalling $\mathcal W_i^2(k)=\alpha_i\|e_i(k)\|^2+\psi_i(k)$ from Lemma~\ref{lem:W_barycenter_decomposition}, Theorem~\ref{thm:uub} further yields the following result.

\begin{theorem}[Ultimate Bound of the Local Wasserstein Tracking Error]
\label{thm:W_ultimate_bound}
Suppose the assumptions of
Theorem~\ref{thm:uub} hold. In addition, suppose that
the local reference spread satisfies
\eqref{eq:uniform_reference_spread}. Then, the local Wasserstein tracking
error generated by the feedforward D$^2$OC controller is uniformly
ultimately bounded. In particular,
\begin{equation}
\limsup_{k\rightarrow\infty}
\mathcal W_i(k)
\leq
\bar{\mathcal W}_i,
\label{eq:W_ultimate_main}
\end{equation}
where
\begin{equation}
\bar{\mathcal W}_i
:=
\sqrt{
\alpha_i
\|C_i\|^2
\eta_i^2
\bar d_i(\xi)^2
+
\bar{\psi}_i
}.
\label{eq:W_explicit_bound}
\end{equation}

Consequently, the local Wasserstein tracking error consists of two
distinct contributions: the closed-loop tracking component induced by
the feedforward mismatch and reference perturbations, and the
irreducible spatial-spread component of the local reference
distribution.
\end{theorem}

\begin{proof}
By Lemma~\ref{lem:W_barycenter_decomposition} and \eqref{eq:uniform_reference_spread},
\begin{equation}
\mathcal W_i^2(k)=\alpha_i\|e_i(k)\|^2+\psi_i(k)\le\alpha_i\|e_i(k)\|^2+\bar{\psi}_i.
\label{eq:W_upper_intermediate}
\end{equation}
By Theorem~\ref{thm:uub}, $\limsup_{k\rightarrow\infty}\|e_i(k)\|\le\|C_i\|\eta_i\bar d_i(\xi)$. Since squaring and the square root are continuous and nondecreasing on $\mathbb R_{\ge0}$, they commute with the limit superior, and taking the limit superior of \eqref{eq:W_upper_intermediate} gives $\limsup_{k\rightarrow\infty}\mathcal W_i(k)\le\sqrt{\alpha_i\|C_i\|^2\eta_i^2\bar d_i(\xi)^2+\bar{\psi}_i}$.
\end{proof}

\begin{corollary}[Vanishing-Spread Case]
\label{cor:W_vanishing_spread}
Under the assumptions of Theorem~\ref{thm:W_ultimate_bound}, suppose
in addition that
$
\bar{\psi}_i=0.
$

Then, we have
\begin{equation}
\limsup_{k\rightarrow\infty}
\mathcal W_i(k)
\leq
\sqrt{\alpha_i}\,
\|C_i\|\eta_i\bar d_i(\xi).
\label{eq:W_zero_spread_bound}
\end{equation}
Specifically, if
$
\bar d_i(\xi)=0,
$
which for $a_i>0$ requires $\xi=1$, $b_i=0$, and $\zeta_i=0$, then
\begin{equation}
\lim_{k\rightarrow\infty}
\mathcal W_i(k)
=
0.
\label{eq:W_convergence_zero}
\end{equation}
\end{corollary}

\begin{proof}
When $\bar{\psi}_i=0$, \eqref{eq:W_explicit_bound} reduces to $\bar{\mathcal W}_i=\sqrt{\alpha_i}\|C_i\|\eta_i\bar d_i(\xi)$. If, in addition, $\bar d_i(\xi)=0$, then $0\le\mathcal W_i(k)$ and $\limsup_{k\rightarrow\infty}\mathcal W_i(k)\le0$ give \eqref{eq:W_convergence_zero}.
\end{proof}

\begin{remark}[Interpretation of the Wasserstein Bound]
Theorem~\ref{thm:W_ultimate_bound} does not claim convergence of the Wasserstein tracking error to zero when the local reference spread is nonzero. The bound has two components: $\alpha_i\|C_i\|^2\eta_i^2\bar d_i(\xi)^2$ quantifies the closed-loop tracking error, and $\bar\psi_i$ is the intrinsic spatial dispersion of the local reference samples around their barycenter. As a result, even perfect barycenter tracking does not imply zero Wasserstein tracking error when the spread is nonzero. 
\end{remark}

\section{Simulation}
\label{sec:simulation}

This section evaluates the proposed feedforward D$^2$OC in simulation. The goals are to show that the feedforward removes the lag, to quantify the dependence on $\xi$, and to support the theory numerically. We compare the original D$^2$OC ($\xi=0$, reference frozen at the current local barycenter) with the proposed feedforward D$^2$OC ($\xi=1$, advected reference), together with intermediate values of $\xi$. Both protocols use the same cost, the same local mass allocation, and the same weights, and thus the only difference is the horizon reference $\bar Q_i^\xi$.

\begin{figure*}[t]
\centering
\includegraphics[width=0.2\textwidth]{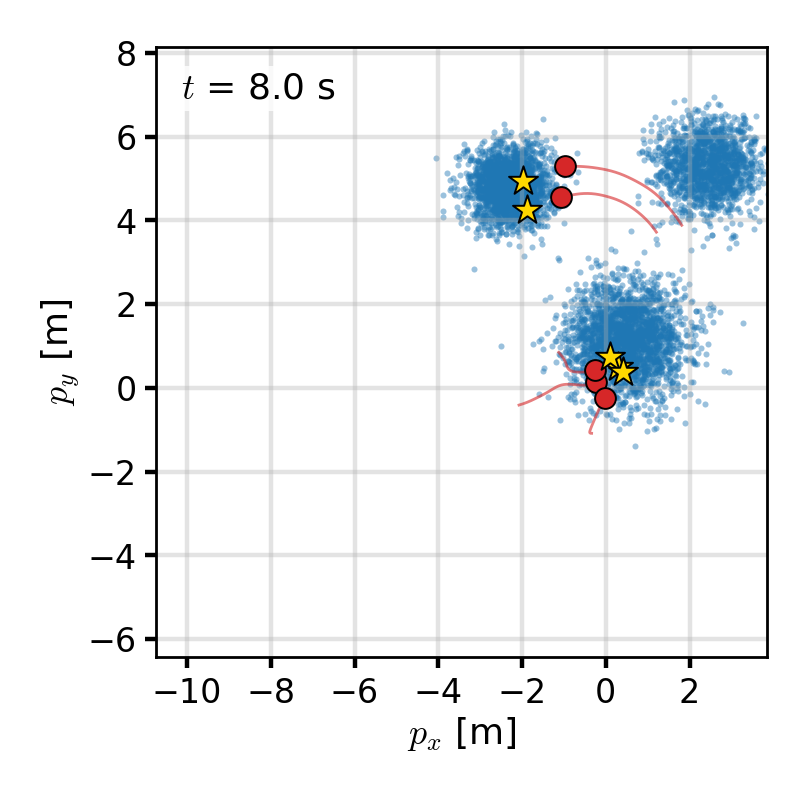}\hfill
\includegraphics[width=0.2\textwidth]{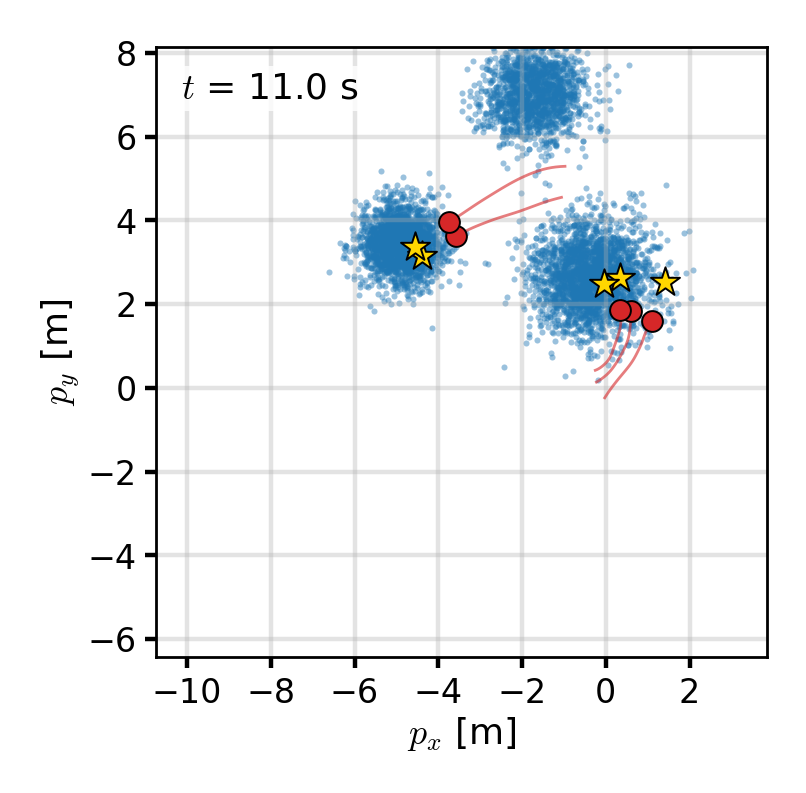}\hfill
\includegraphics[width=0.2\textwidth]{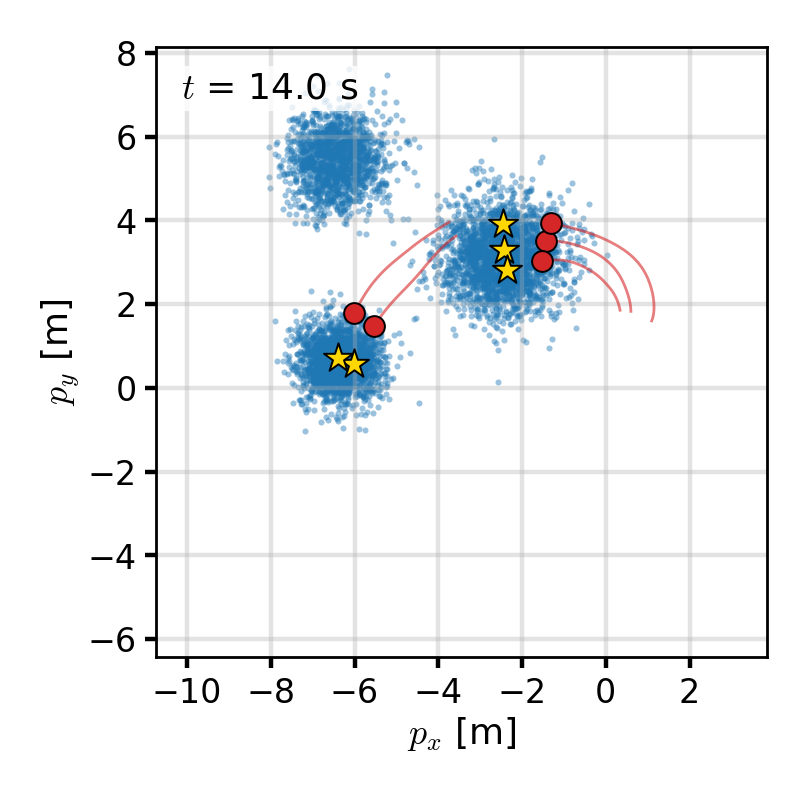}\hfill
\includegraphics[width=0.2\textwidth]{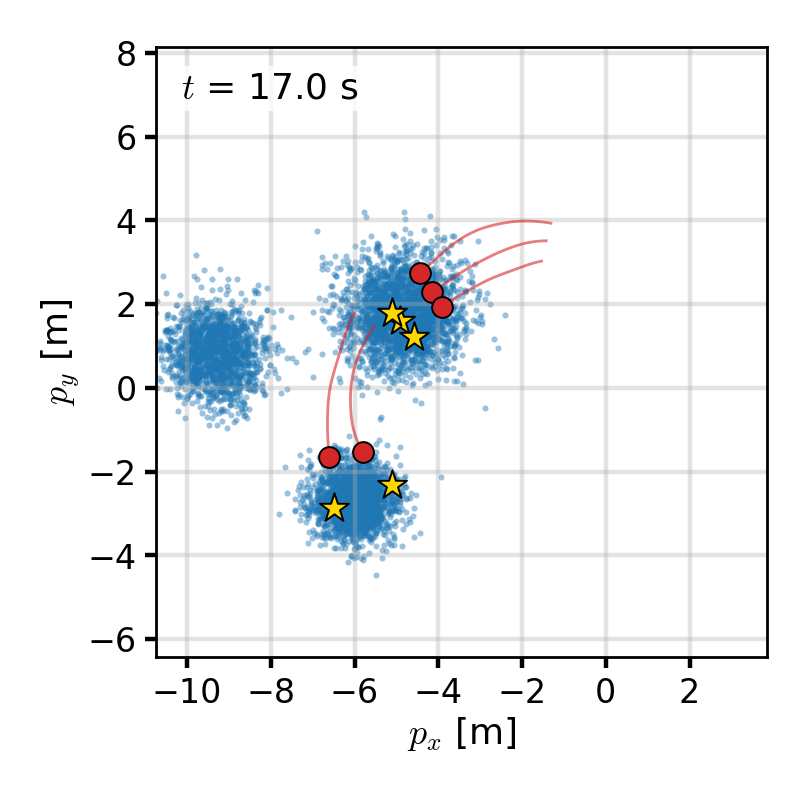}\hfill
\includegraphics[width=0.2\textwidth]{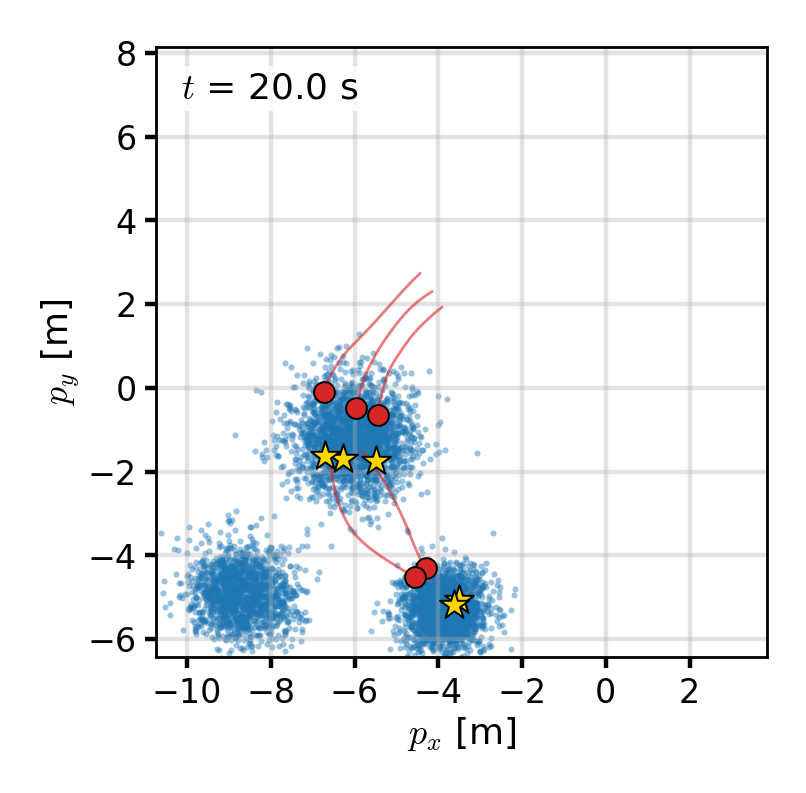}\\[2pt]
\includegraphics[width=0.2\textwidth]{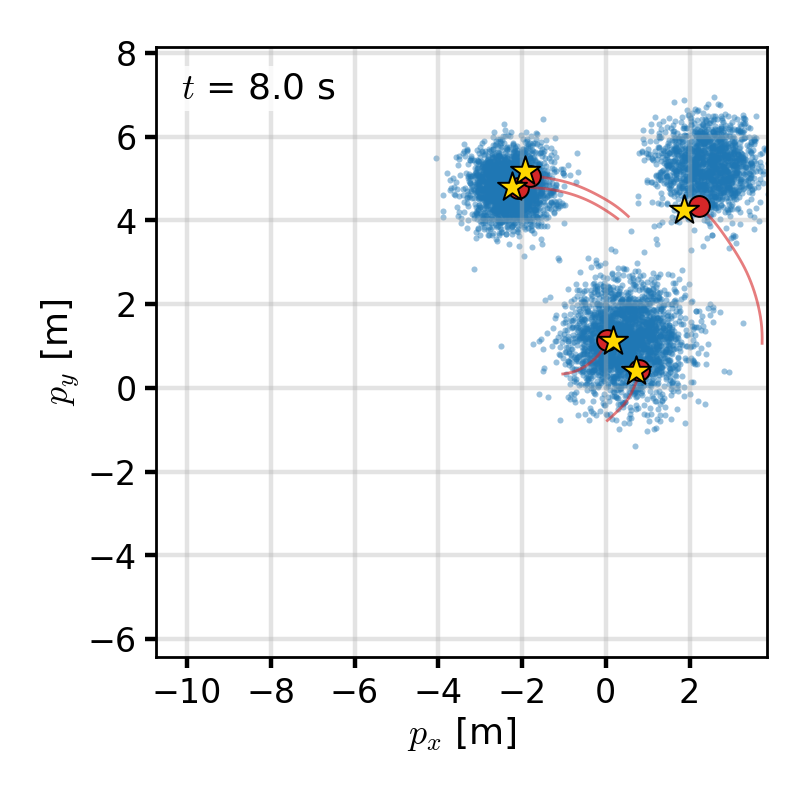}\hfill
\includegraphics[width=0.2\textwidth]{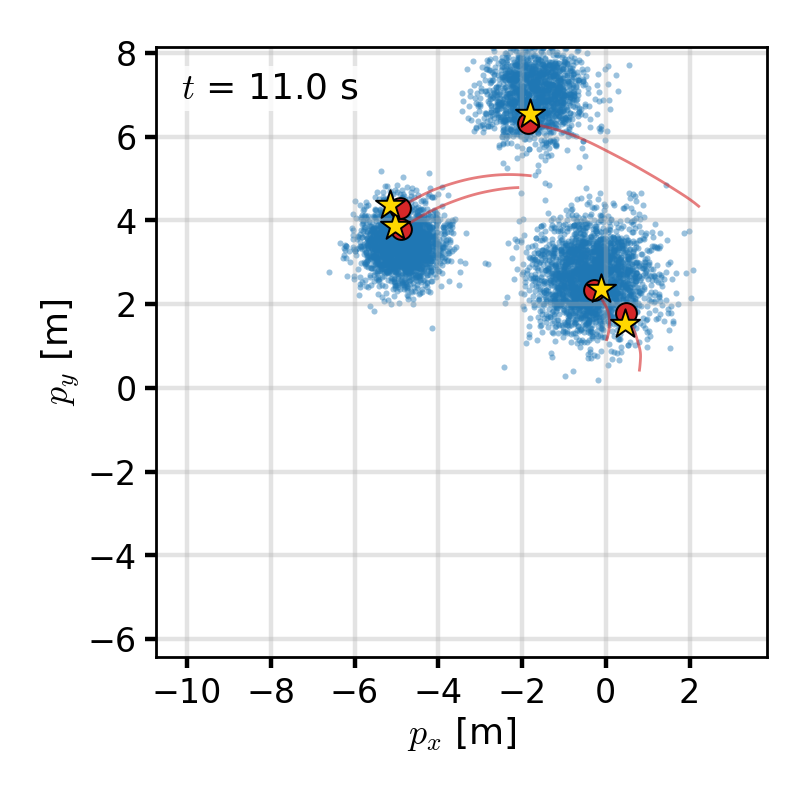}\hfill
\includegraphics[width=0.2\textwidth]{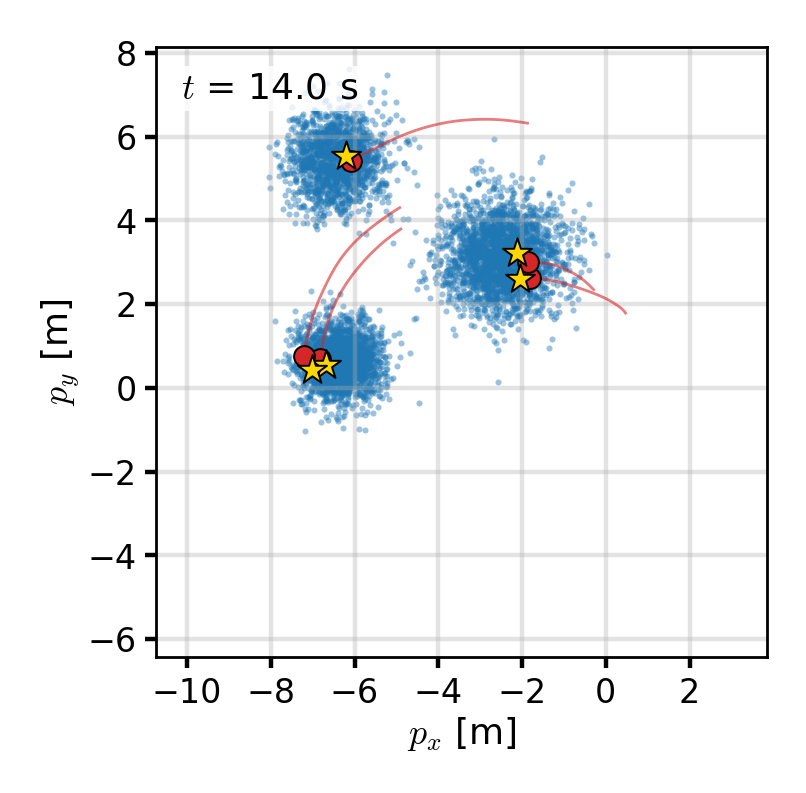}\hfill
\includegraphics[width=0.2\textwidth]{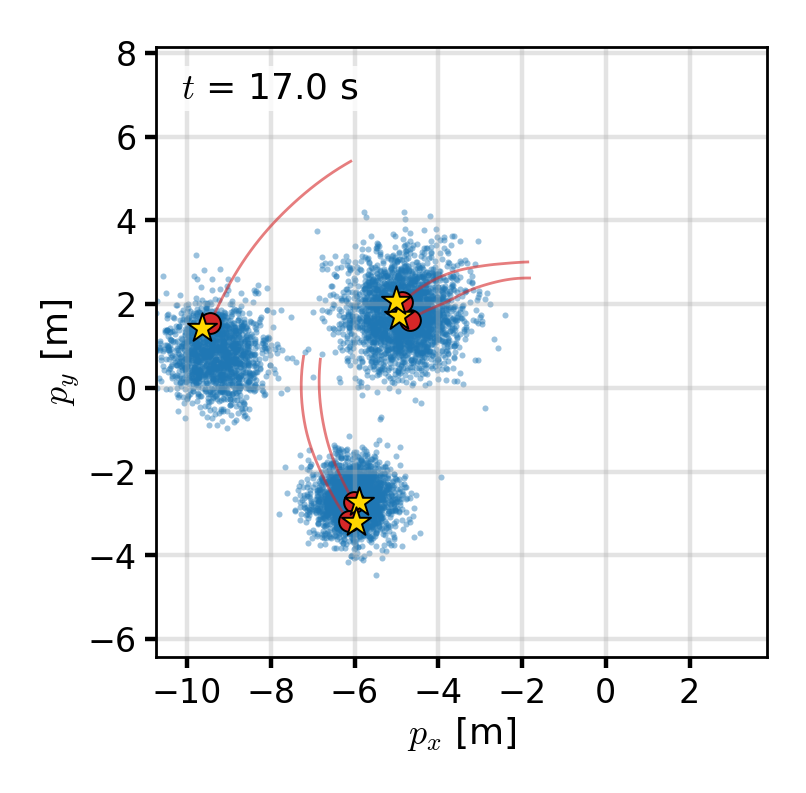}\hfill
\includegraphics[width=0.20\textwidth]{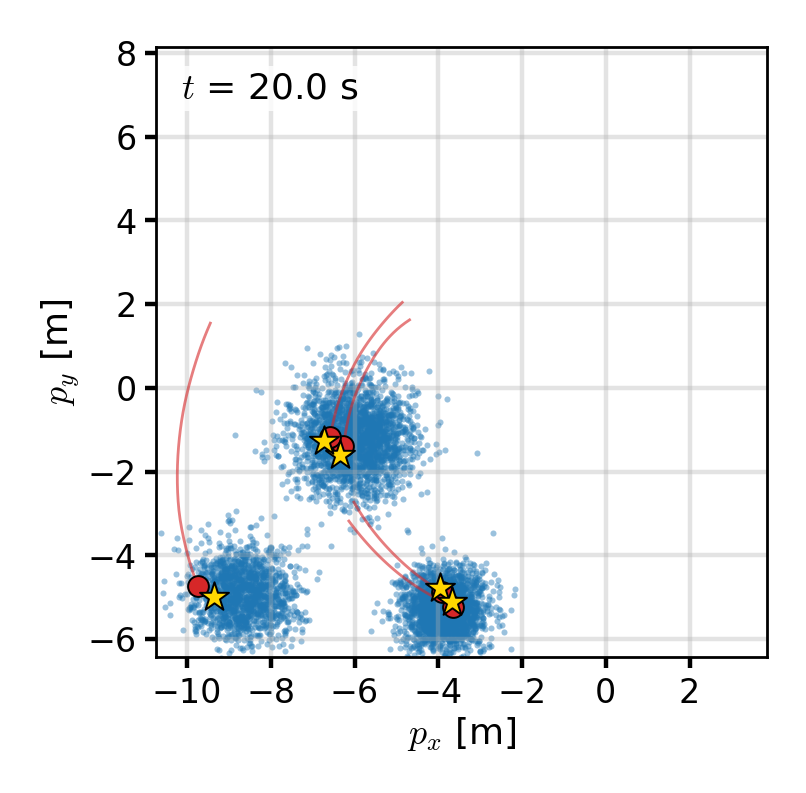}
\caption{Snapshots at $t=8,\,11,\,14,\,17,\,20$ s. Top row: original D$^2$OC ($\xi=0$). Bottom row: proposed feedforward D$^2$OC ($\xi=1$). Blue dots are the reference samples (identical in both rows), red circles are the agents with their last 3 s of trajectory, and stars are the local barycenters $\bar q_i(k)$.}
\label{fig:snapshots}
\end{figure*}
\subsection{Simulation Setup}
The swarm consists of $n_a=5$ identical quadrotors, each linearized at hover and moving in the horizontal plane at constant altitude, with state $x=[p_x,\,p_y,\,\theta,\,\phi,\,\dot p_x,\,\dot p_y,\,\dot\theta,\,\dot\phi]^{\top}\in\R^8$, where $\theta$ and $\phi$ are the pitch and roll angles, input $u=[\tau_\theta/J_\theta,\, \tau_\phi/J_\phi]^{\top}\in\R^2$ (pitch and roll torques per inertia), and output $y=[p_x,\;p_y]^{\top}\in\R^2$. The linearized model is given by
\begin{equation*}
\ddot p_x=-g\,\theta,\;\; \ddot p_y=g\,\phi,\;\; \ddot\theta=u_1,\;\; \ddot\phi=u_2,\;\; g=9.81\ \mathrm{m/s^2}.
\end{equation*}
The model is discretized by zero-order hold with $\Delta t=0.1$ s. It is controllable, satisfies the rank condition of Assumption~\ref{ass:standing} (smallest singular value $0.097$), and has $U_{e,i}=\mathbf 0$ because no torque is needed to hold a position. The matrix $A_i$ has all eigenvalues on the unit circle, and thus the agent is not asymptotically stable without any control. The simulation uses this linear model for both the controller and the agent.

The target distribution is a mixture of three Gaussians with $N=6000$ samples. The component means are $(-2.5,1)$, $(2,2)$, and $(0.5,-2.5)$ m, the standard deviations are $0.7$, $0.5$, and $0.6$ m, and the mixture weights are $0.4$, $0.35$, and $0.25$, so that the initial distribution fits in a domain of about $10\ \mathrm{m}\times10\ \mathrm{m}$. The samples are advected by the known velocity field
\begin{equation*}
\begin{aligned}
v(x,t)&=\mathrm{sat}_{\bar v}\!\Big(0.5\,[\cos(0.12t),\ \sin(0.12t)]^\top\\
&\quad +0.25\,[-x_2,\ x_1]^\top\Big),
\quad\bar v=10\ \mathrm{m/s},
\end{aligned}
\end{equation*}
where $\mathrm{sat}_{\bar v}(w):=\bar v\tanh(\|w\|/\bar v)\,w/\|w\|$ saturates the speed smoothly at $\bar v$. The field is a rigid rotation superposed on a slowly circling translation and the Gaussians keep their shape. The true samples are propagated by a fourth-order Runge--Kutta method, whereas the controller predicts them with the forward Euler step of \eqref{eq:sample_vf}.

We set the horizon $T=10$, local mass $\alpha_i=1/M$ with mission length $M=1800$ steps ($180$ s), $R_i=\alpha_i r\mathbf I$ with $r=0.02$, and $P_i$ solves the Riccati equation \eqref{eq:DARE} with $R_i=\alpha_i r\mathbf I$ and $Q_i=\alpha_i(q_pC_i^\top C_i+q_v\mathbf I)$, $q_p=0.1$, $q_v=1$ (equivalently, $P_i=\alpha_iP^\mathrm{are}_i$ with $P^\mathrm{are}_i$ the solution for the weights $r\mathbf I$ and $q_pC_i^\top C_i+q_v\mathbf I$). The nominal closed loop $A_{\mathrm{cl},i}=A_i-B_iK_i^{\mathrm{mpc}}$ is Schur stable with $\rho(A_{\mathrm{cl},i})=0.851$, and thus Assumption~\ref{ass:nominal} holds. The local selection is the greedy allocation of mass to the nearest samples, kept fixed over the horizon. Each agent starts at rest at a uniformly random position in $[-4,4]^2$.
 
Each agent maintains its own weight map of the remaining target density, initialized uniformly and depleted by the local mass $\alpha_i=1/M$ allocated at each step, so that an isolated agent would exhaust its capacity after exactly $M$ steps. Agents within the communication range $d_{\mathrm{comm}}=3$ m take the elementwise minimum of the sample weights over themselves and their neighbors. Each map is therefore also depleted by the mass that neighbors have already covered, which prevents redundant coverage but makes the maps empty much faster than in the isolated case. The mission terminates when the first agent exhausts its weight capacity, and thus well before $M$ steps.

Performance metrics are evaluated only while every agent retains at least $30\%$ of its weight capacity. Within each trial, all protocols are evaluated on the same window, namely the shortest one among all values of $\xi$. The window is $26.5$ s in the representative run and $26.9$ s on average over the Monte Carlo trials (between $25.7$ and $29.9$ s), compared with the nominal mission length $M\Delta t=180$ s.
 
The tracking error is $\|e_i(k)\|=\|y_i(k)-\bar q_i(k)\|$, and the local Wasserstein distance is $\mathcal W_i(k)$ in \eqref{eq:local_W_definition}, normalized as $\hat W_i:=\mathcal W_i/\sqrt{\alpha_i}$. The overbar symbols denote averages over all agents and over the steps after the transient ($k\ge 60$, i.e., $t\ge 6$ s) within the evaluation window. We run 30 Monte Carlo trials, each with a new draw of the reference samples and of the initial agent positions, and all values of $\xi$ are run on the same draw. We report the mean over the trials of each overbar quantity with the half-width of its 95\% confidence interval, and, because the trials are paired, the mean difference $\Delta\overline{\|e\|}(\xi)$ between $\xi$ and $\xi=0$ within the same trial, also with a 95\% confidence interval. We additionally report the $p$-value of the Wilcoxon signed-rank test \cite{hollander2013nonparametric} applied to the 30 differences, a nonparametric test of whether the 30 paired differences are systematically nonzero.
 
\subsection{Tracking Behavior}
Fig.~\ref{fig:snapshots} shows snapshots of both protocols at the same times, for the same reference realization. With the original D$^2$OC, the local barycenters (stars) run ahead of the agents (circles), which trail behind the moving reference. On the other hand, with the proposed method the agents stay on top of their barycenters, and the part of the target they cover remains aligned with the moving distribution. Fig.~\ref{fig:compare}(a) shows the same effect in the tracking error. The error of the original method grows from about $0.6$ at $5$ s to above $2.1$ at $26$ s, while that of the proposed method stays around $0.25$.
The attitude angles of the linearized quadrotor remain small during the evaluation window. The 99th percentile after the transient does not exceed $8.1^\circ$ for all $\xi$, within the range in which the small-angle linearization is valid. The worst single value over all $\xi$ and the whole run is $19.6^\circ$.

\begin{figure}[t]
\centering
\subfloat[]{\includegraphics[width=0.49\columnwidth]{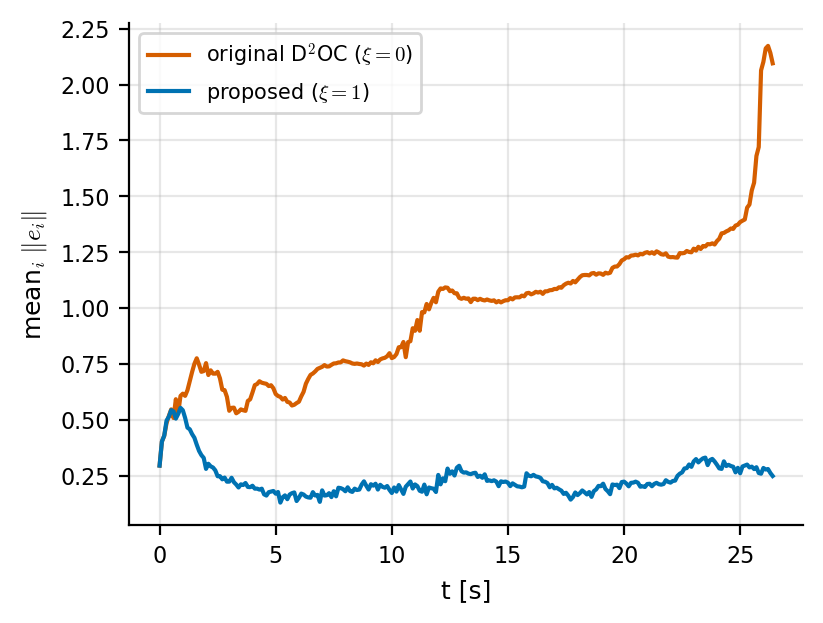}}
\subfloat[]{\includegraphics[width=0.49\columnwidth]{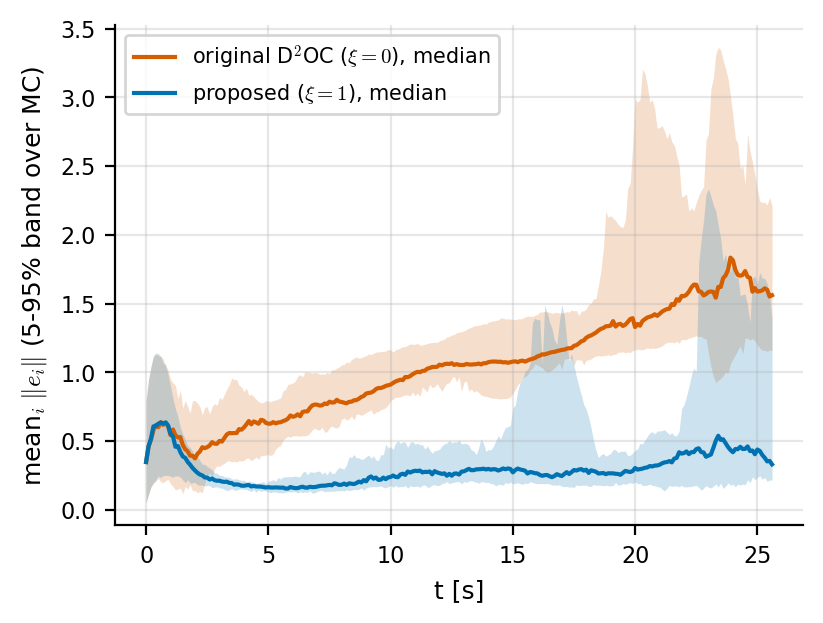}}
\caption{Tracking error over time. (a) Mean tracking error of the five agents in the representative run for the original ($\xi=0$) and the proposed ($\xi=1$) D$^2$OC. (b) Median and 5--95\% band of the mean tracking error over the 30 Monte Carlo trials.}
\label{fig:compare}
\end{figure}

\subsection{Statistics over $\xi$}
Fig.~\ref{fig:compare}(b) shows that the improvement holds over time. After the first two seconds, the median error of the proposed method stays between $0.15$ and $0.55$ until $25$ s, whereas that of the original method grows from about $0.6$ at $5$ s to about $1.6$ at $25$ s. The 5--95\% bands are separated during most of this interval and overlap in some parts (around $15$--$18$ s and from about $23$ s to the end of the window), where individual trials show error spikes.

Table~\ref{tab:mc} and Fig.~\ref{fig:mc_xi}(a) provide the Monte Carlo results. The mean error and the mean local Wasserstein distance decrease monotonically with $\xi$ for this simulation, from $1.266$ at $\xi=0$ to $0.385$ at $\xi=1$ for the error, a reduction of $69.6\%$. In every one of the 30 trials, each $\xi>0$ gives a smaller error than $\xi=0$ on the same scenario. If $\xi$ had no effect, the chance of all 30 differences having the same sign would be $2/2^{30}\approx1.9\times10^{-9}$, which is exactly the $p$-value of the Wilcoxon signed-rank test and the smallest value attainable with 30 pairs. The reduction is therefore systematic and not due to chance. Its size is given by $\Delta\|e\|$ in Table~\ref{tab:mc}.

\begin{table}[t]
\centering
\caption{Monte Carlo results over 30 trials (mean $\pm$ 95\% CI; paired difference with respect to $\xi=0$). $\overline{\|u\|}$ is the mean input norm.}
\label{tab:mc}
\setlength{\tabcolsep}{4pt}
\begin{tabular}{ccccc}
\hline
$\xi$ & $\overline{\|e\|}$ & $\overline{\hat W}$ & $\Delta\overline{\|e\|}$ vs.\ $\xi=0$ & $\overline{\|u\|}$ \\
\hline
0    & $1.266\pm0.047$ & $1.285\pm0.047$ & --                & $1.07\pm0.06$ \\
0.25 & $1.024\pm0.035$ & $1.046\pm0.035$ & $-0.242\pm0.029$ & $1.05\pm0.06$ \\
0.5  & $0.801\pm0.031$ & $0.828\pm0.032$ & $-0.465\pm0.040$ & $1.07\pm0.07$ \\
0.75 & $0.561\pm0.028$ & $0.593\pm0.028$ & $-0.704\pm0.045$ & $1.00\pm0.07$ \\
1    & $0.385\pm0.033$ & $0.428\pm0.034$ & $-0.881\pm0.052$ & $1.03\pm0.09$ \\
\hline
\end{tabular}
\end{table}
 
\begin{figure}[t]
\centering
\subfloat[]{\includegraphics[width=0.49\columnwidth]{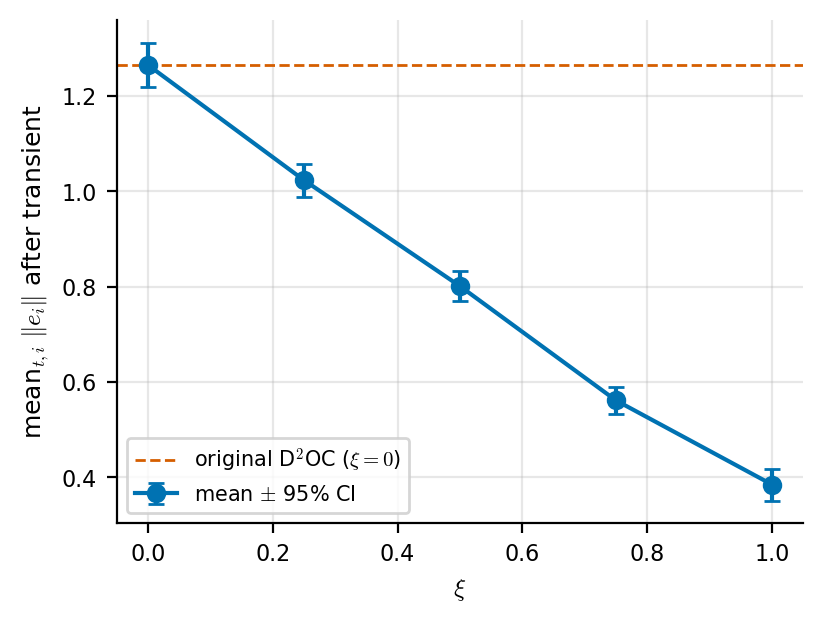}}
\subfloat[]{\includegraphics[width=0.49\columnwidth]{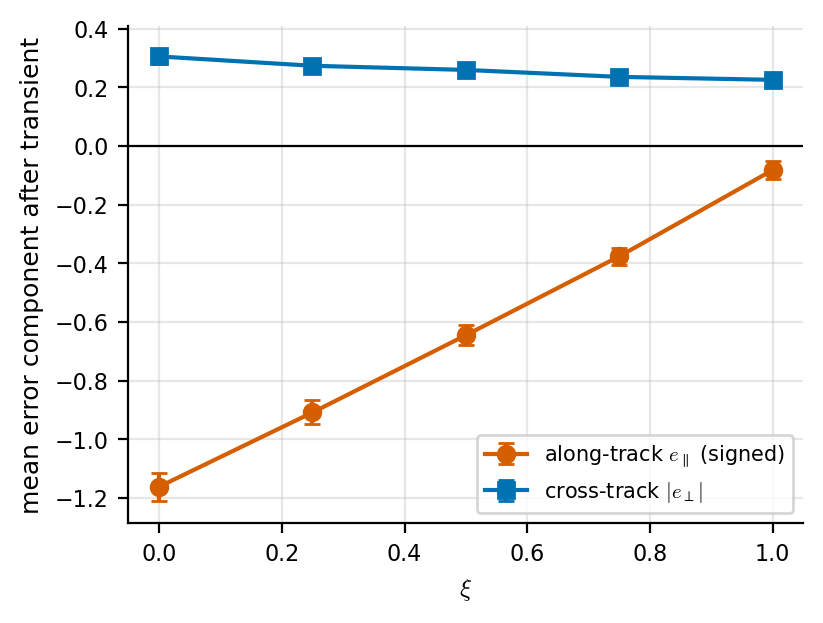}}
\caption{(a) Mean tracking error after the transient versus $\xi$ (dashed: original D$^2$OC). (b) Mean along-track error $e_\parallel$ (signed, negative means lagging) and mean cross-track error $|e_\perp|$ versus $\xi$. 95\% confidence intervals over 30 Monte Carlo trials.}
\label{fig:mc_xi}
\end{figure}

\subsection{Lag, Input Effort, Reference Speed, and Velocity-Model Error}
To show that the performance improvement is the removal of the lag, we split the error into the component along the predicted drift, $e_\parallel=e_i^\top\mathbf t_i$ with $\mathbf t_i=\Delta\bar q_{i,1}/\|\Delta\bar q_{i,1}\|$, and the orthogonal component $e_\perp$. A negative $e_\parallel$ means that the agent is behind the barycenter. Fig.~\ref{fig:mc_xi}(b) shows that the original D$^2$OC has a mean along-track error of $-1.16\pm0.05$, and thus its error is mostly a lag, and that the magnitude of the along-track error decreases approximately linearly in $\xi$, from $1.16\pm0.05$ at $\xi=0$ to $0.08\pm0.03$ at $\xi=1$. The cross-track error decreases only from $0.31$ to $0.23$. Interestingly, the reduction of the error does not come at the cost of a larger input (see Table~\ref{tab:mc}). The mean input norm is $1.07\pm0.06$ for $\xi=0$ and $1.03\pm0.09$ for $\xi=1$, and the confidence intervals overlap.

Theorem~\ref{thm:lag_reduction}(i) states that the cost of the lag grows with the predicted drift. We scale the speed of the reference flow and run 10 trials for each speed as shown in Fig.~\ref{fig:sweeps}(a). The mean error of the original D$^2$OC grows from $0.49$ to $1.28$ when the mean barycenter speed increases from $0.23$ to $1.13$ m/s, whereas the error of the proposed method stays between $0.34$ and $0.39$. The reduction increases from $29.8\%$ at the lowest speed to about $70\%$ at the highest speeds ($70.7\%$ at $1.13$ m/s and $68.7\%$ at $1.07$ m/s).

Fig.~\ref{fig:sweeps}(b) evaluates the robustness to an inexact velocity field. The controller predicts trajectories using a perturbed field $(1-\epsilon)v$ over 10 trials for each value of $\epsilon$, where $\epsilon$ indicates a velocity model error. Under the exact model ($\epsilon=0$), the tracking error ratio relative to the original D$^2$OC is $0.30$ for $\xi=1$ and $0.63$ for $\xi=0.5$. When the velocity field is underestimated ($\epsilon > 0$), the performance benefit gradually diminishes as $\epsilon$ approaches $1$. Specifically, the error ratios for $\xi=1$ are $0.44$, $0.62$, and $0.79$ at $\epsilon=0.25$, $0.5$, and $0.75$, respectively, and they reach $1$ at $\epsilon=1$, where no velocity information is utilized and the method coincides with the original D$^2$OC. Conversely, when the field is overestimated by $50\%$ ($\epsilon=-0.5$), full feedforward yields a ratio of $0.57$, which is higher than the $0.44$ achieved by half feedforward ($\xi=0.5$), yet remains superior to the baseline ($\xi=0$). This observation aligns with Remark~\ref{remark:lag_reduction}, which states that the value of $\xi$ minimizing the realized error depends on prediction accuracy. The performance benefit decreases as the mismatch between the predicted and true drift grows, indicating that a partial compensation ($\xi < 1$) is preferable when the prediction overestimates the motion.
 
\begin{figure}[t]
\centering
\subfloat[]{
\includegraphics[width=0.49\columnwidth]{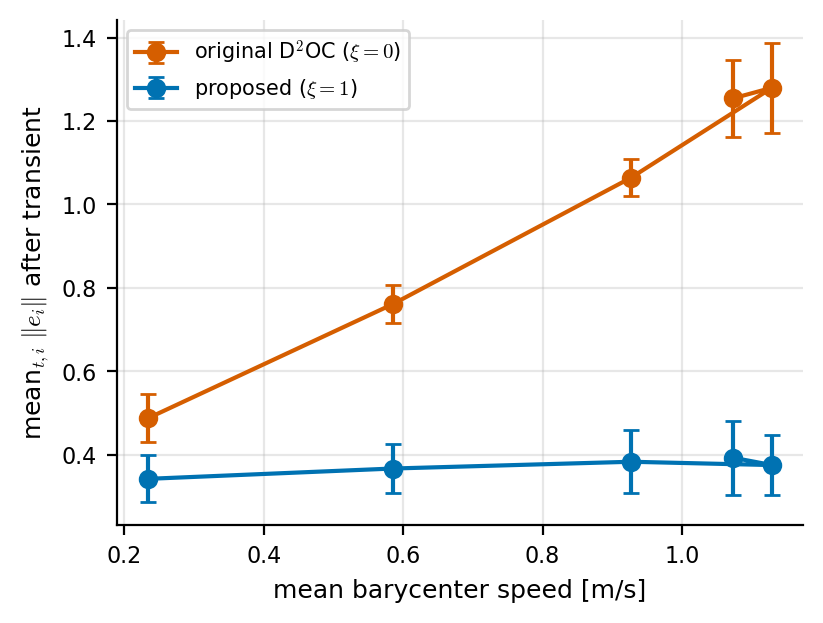}}
\subfloat[]{
\includegraphics[width=0.49\columnwidth]{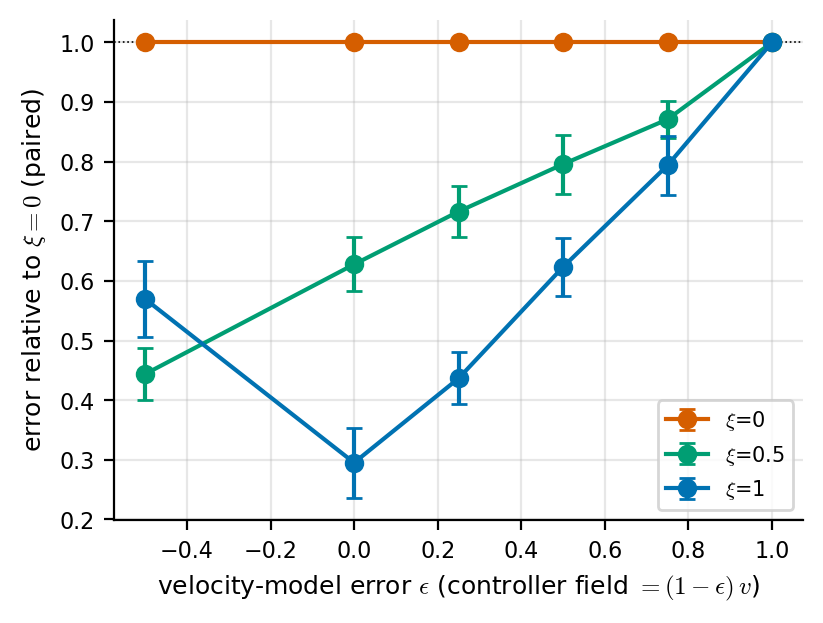}
}
\caption{(a) Mean tracking error after the transient versus the mean barycenter speed. (b) Error relative to $\xi=0$ in the same trial versus the velocity-model error $\epsilon$. Both with 95\% confidence intervals over 10 trials.}
\label{fig:sweeps}
\end{figure}
 
\subsection{Verification of the Theorems}

Figs.~\ref{fig:thm3}(a) and (b) verify the closed-loop guarantees. In all runs, the largest tracking error among the agents remains below the square root of the bound \eqref{eq:x_transient} of Theorem~\ref{thm:uub} multiplied by $\|C_i\|$, and the largest local Wasserstein distance remains below the bound of Theorem~\ref{thm:W_ultimate_bound}. Consequently, the tracking error and the local Wasserstein distance are consistent with the ultimate bounds for every tested value of $\xi$, including $\xi=1$.

\begin{figure}[t]
\centering
\subfloat[]{\includegraphics[width=0.49\columnwidth]{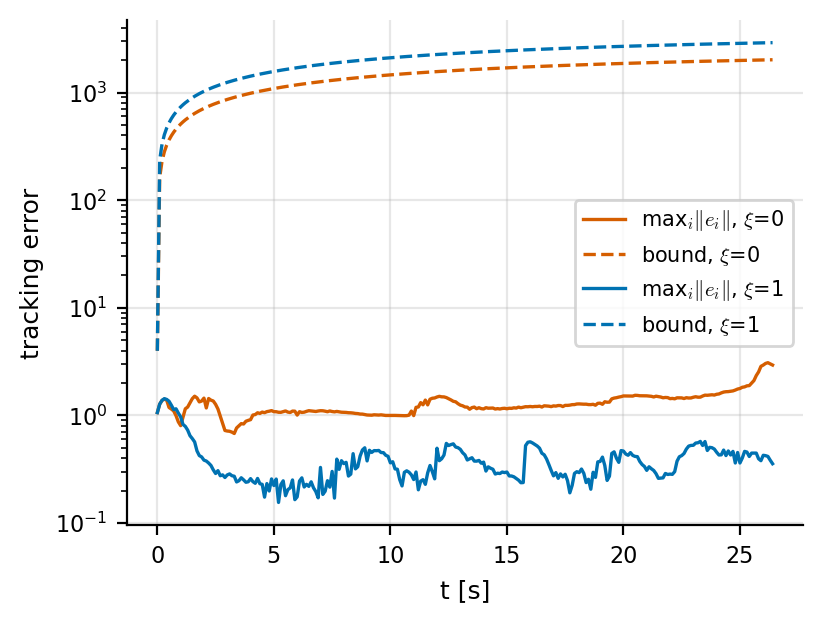}}
\subfloat[]{\includegraphics[width=0.49\columnwidth]{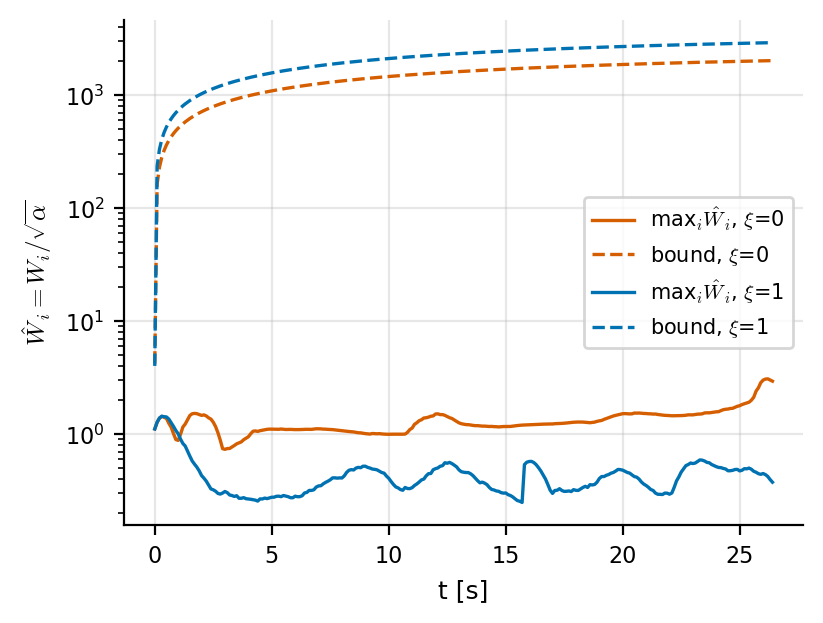}}\\
\subfloat[]{\includegraphics[width=0.55\columnwidth]{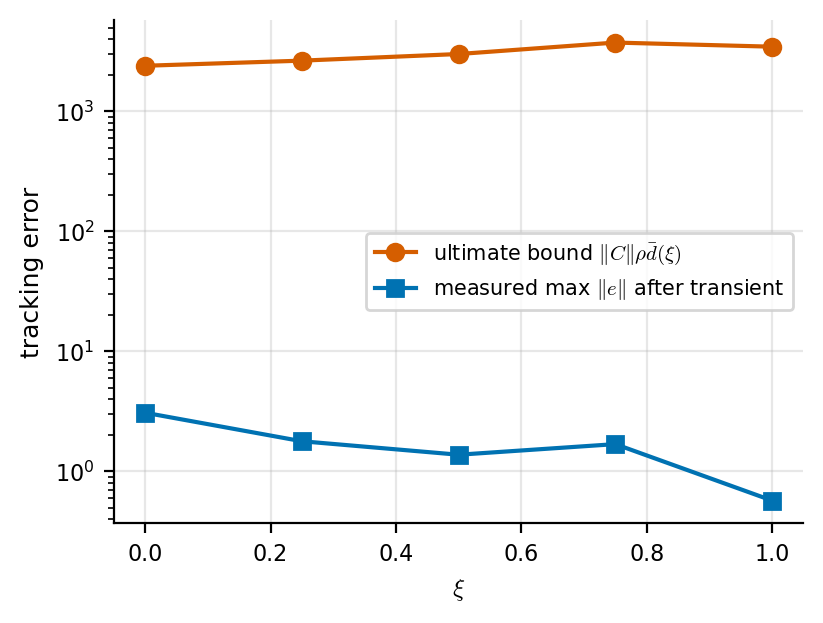}}
\caption{Verification of the bounds. (a) Theorem~\ref{thm:uub}: maximum tracking error over the agents (solid) and the bound (dashed) versus time. (b)
Theorem~\ref{thm:W_ultimate_bound}: maximum normalized local Wasserstein distance $\hat W_i=\mathcal W_i/\sqrt{\alpha_i}$ over the agents (solid) and the bound (dashed).
(c) Theorem~\ref{thm:uub}: ultimate bound and measured maximum error after the transient versus $\xi$.}
\label{fig:thm3}
\end{figure}

The bounds are conservative by more than two orders of magnitude. The constants are computed with $Q_{L,i}=\mathbf I$ and with $\epsilon_i$ chosen by a grid search that minimizes $\eta_i$, and the resulting Lyapunov constant $\eta=1.8\times10^3$ is large for this poorly damped agent. The largest measured error after the transient is $3.1$ for $\xi=0$ and $0.57$ for $\xi=1$, whereas the corresponding ultimate bounds are $2.4\times10^3$ and $3.5\times10^3$. The bounds therefore certify boundedness but do not predict the achieved performance. 
In particular, the bound is larger for $\xi=1$ although the measured error is smaller (Fig.~\ref{fig:thm3}(c)). The reason is that $a=0.23$ and $b\approx1.38$ for this agent, so that the condition $b_i<a_i$ for a lower bound is not satisfied, and that the perturbation constant $\zeta_i$, which is evaluated from the realized reselection jitter, differs between the runs ($\zeta=1.19$ for $\xi=0$ and $0.59$ for $\xi=1$). The error reduction reported above is thus attributed to the horizon cost of Theorem~\ref{thm:lag_reduction}(i), and not to the ultimate bound.

\section{Conclusion}

This paper presented a feedforward extension of D$^2$OC that reduces the systematic tracking lag of stationary-reference protocols when following dynamic target distributions. By advecting the horizon reference along the target velocity field, the framework yields a closed-form receding-horizon control law that embeds motion prediction directly into the finite-horizon optimization. The analysis shows that the feedforward reduces the horizon cost of the stationary protocol, whose excess vanishes at full compensation when the predictions are exact, and that the tracking error and the local Wasserstein distance remain ultimately bounded, with explicit bounds that depend on the compensation level. Simulations with quadrotor dynamics show a substantial reduction of the tracking error and a smaller benefit under velocity-model mismatch. Future work will explore input constraints, tighter bounds, and global distribution-matching guarantees.

	\bibliographystyle{IEEEtran}
	\bibliography{references}

\end{document}